\documentclass{article}

\usepackage{times}
\usepackage[pdftex]{graphicx}
\usepackage{amssymb,amsfonts,amsmath,amsthm}
\usepackage{float}
\usepackage{xcolor}
\usepackage{tikz}
\tikzset{
  treenode/.style = {shape=rectangle, rounded corners,
                     draw, align=center,
                     top color=white, bottom color=white},
  root/.style     = {treenode, font=\Large, bottom color=white},
  env/.style      = {treenode, font=\ttfamily\normalsize},
  dummy/.style    = {circle,draw}
}
\newtheorem{theorem}{Theorem}
\newtheorem{proposition}{Proposition}
\newtheorem{lemma}{Lemma}
\newtheorem{ass}{Assumption}
\newtheorem{definition}{Definition}
\newtheorem{remark}{Remark}
\newtheorem{cor}{Corollary}

\newcommand\EE {\mathbb E}
\newcommand\FF {\mathbb F}

\newcommand\RR {\mathbb R}
\newcommand\PP {\mathbb P}
\newcommand\QQ {\mathbb Q}

\def\bone{\mathbf{1}}

\def\qed{\hskip6pt\vrule height6pt width5pt depth1pt}

\def\qed{\hskip 6pt\vrule height6pt width5pt depth1pt}

\newcommand{\ed}{\end{document}}
\newcommand{\be}{\begin{equation}}
\newcommand{\ee}{\end{equation}}
\newcommand{\bq}{\begin{eqnarray}}
\newcommand{\eq}{\end{eqnarray}}

\vspace{4in}

\newcommand{\supp}{\text{supp}}

\begin{document} 

\title{Endogenous Formation of Limit Order Books: Dynamics Between Trades.\footnote{Partial support from the NSF grant DMS-1411824 is acknowledged by both authors.}}
\author{Roman~Gayduk and Sergey~Nadtochiy\footnote{Address the correspondence to: Mathematics Department, University of Michigan, 530 Church St, Ann Arbor, MI 48104; sergeyn@umich.edu.}\footnote{We would like to thank the anonymous referees whose constructive remarks helped us improve the paper.}}
\date{Current version: June 19, 2017\\
Original version: May 26, 2016
}
\maketitle

\begin{abstract}

In this work, we present a continuous-time large-population game for modeling market microstructure between two consecutive trades. The proposed modeling framework is inspired by our previous work \cite{GaydukNadtochiy1}.
In this framework, the Limit Order Book (LOB) arises as an outcome of an equilibrium between multiple agents who have different beliefs about the future demand for the asset. The agents' beliefs may change according to the information they observe, triggering changes in their behavior.
We present an example illustrating how the proposed models can be used to quantify the consequences of changes in relevant information signals.
If these signals, themselves, depend on the LOB, then, our approach allows one to model the ``indirect" market impact (as opposed to the ``direct" impact that a market order makes on the LOB, by eliminating certain limit orders).
On the mathematical side, we formulate the proposed modeling framework as a continuum-player control-stopping game. We manage to split the equilibrium problem into two parts. The first one is described by a two-dimensional system of Reflected Backward Stochastic Differential Equations (RBSDEs), whose solution components reflect against each other. The second one leads to an infinite-dimensional fixed-point problem for a discontinuous mapping. Both problems are non-standard, and we prove the existence of their solutions in the paper.

\end{abstract}

\section{Introduction}

In this paper, we continue the development of an equilibrium-based modeling framework for market microstructure, initiated in \cite{GaydukNadtochiy1}. 
As in \cite{GaydukNadtochiy1}, we analyze the market microstructure in the context of an auction-style exchange (as most modern exchanges are), in which the participating agents can post limit or market orders.
A crucial component of such a market is the \emph{Limit Order Book (LOB)}, which contains all outstanding limit buy and sell orders (time and price prioritized), and whose shape and dynamics represent the liquidity of the market.
We are interested in developing a modeling framework in which the shape of the LOB, and its dynamics, arise \emph{endogenously} from the interactions between the agents. This is in contrast to many of the existing results on market microstructure, which assume that the shape and dynamics of the LOB are given exogenously. Among the many advantages of our approach is the possibility of modeling the reaction of the LOB to the changes in a relevant market indicator or in the rules of the exchange.\footnote{We refer the reader to \cite{GaydukNadtochiy1}, whose introduction contains a more detailed explanation of the problems of market microstructure and a motivation for our study.}

Herein, we extend the discrete time modeling framework proposed in \cite{GaydukNadtochiy1} to continuous time, and restrict our analysis to the dynamics of the market between two consecutive trades. The latter simplifies the problem and is justified by the well known empirical fact that most changes in LOB are not due to trades.
We manage to establish the existence, and obtain a numerically tractable representation, of an equilibrium in a general continuous time framework, in which the competing agents have different beliefs about the future demand for the asset. These beliefs determine the future distribution of the demand, given the (common) information observed thus far. The latter may, e.g., be generated by a relevant signal (or, market indicator). One can view such conditional distributions as the ``models" that the market participants use to predict future demand, and which are based on the (commonly observed) relevant market indicators.
Given the beliefs, the agents choose their optimal trading strategies (i.e. limit and market orders), aiming to maximize their expected profits, and reach an equilibrium.
The modeling framework proposed herein can be used for predicting the reaction of a market to various changes in the relevant indicators. In particular, if the relevant market indicator depends on the LOB, our framework allows one to model the indirect market impact: i.e. how an initial change to the market may cause further changes to it, due to the information revealed by the initial change (as opposed to the direct impact, e.g., made by a market order eliminating a part of the LOB).
An extreme example of such indirect impact is called ``spoofing", and it is an illegal activity aimed at manipulating the market. Our model can be used to quantify such indirect market impact, and it can be, ultimately, used to improve the optimal execution algorithms or to test the consequences of ``spoofing" activity. We provide a simplistic example illustrating the potential applications of our model in Section \ref{se:example}, although an empirical investigation (including a more careful model specification, and its estimation), which is needed to make any specific conclusions about the actual market behavior, is left for future research.

On the mathematical side, the problem we analyze is the construction of an equilibrium in a control-stopping game with a continuum of players (cf. \cite{Aumann}, \cite{Schmeidler}, \cite{GCarmona}, for more on the general theory of continuum-player games). The main mathematical challenges stem from three sources: the complicated dependence between the individual payoffs and the controls of other players (which lacks the standard convexity and continuity properties), the presence of multiple participants (as compared to a two-player game) and the control-stopping nature of the game. Equilibria in the games with any number of players can often be constructed directly, by means of a system of Partial Differentia Equations or a system of (Forward-) Backward Stochastic Differential Equations (BSDEs). However, in the case of multiple players, solving such systems numerically becomes very challenging. In such cases, the description of an equilibrium is, typically, limited to the proof of its existence, which, in turn, is obtained by an abstract fixed-point argument. However, even the latter method presents a challenge in the game considered herein. Namely, the complicated dependence structure between the players' controls and individual payoffs, along with the control-stopping nature of the game, make it very challenging (or even impossible) to (a) to find a compact set of individual controls, which is sufficiently large to include any maximizer of the objective function, and (b) establish the continuity of the objective.\footnote{Alternatively, one can exploit the monotonicity properties of the objective, to apply a different type of fixed-point theorem. Nevertheless, such monotonicity is also lacking in the present setting.} In order to overcome these challenges, we assume the existence of agents with ``extremal beliefs" to split the problem into two parts: a control-stopping game with two players, and a pure control game (without stopping) with a continuum of players. Such a split simplifies our task dramatically, but both resulting problems remain challenging. The first one, concerned with the construction of an equilibrium in a two-player game, leads to a non-standard system of Reflected BSDEs (RBSDEs), whose components \emph{reflect against each other}, and whose generator \emph{lacks to desired regularity}. In Subsection \ref{subse:RBSDE.exist}, we prove the existence of a solution to this system, and, in Section \ref{se:example}, we show how it can be computed in a simple example.
The second problem, concerned with the equilibrium in a continuum-player game (without stopping), is formulated as a fixed-point problem, and is solved in Subsection \ref{subse:fixedpoint}.
This auxiliary game is complicated by the fact that it has a \emph{discontinuous} objective function and does not possess the desired monotonicity properties. Nevertheless, an appropriate ``mollification" technique is designed in Subsection \ref{subse:fixedpoint} to construct a solution to the associated fixed-point problem, and, in turn, to describe an equilibrium in the original market microstructure game.
One of the computational benefits of the solution method proposed herein is that the aforementioned fixed-point problem can be solved separately for each $(t,\omega)$. In particular, it is not necessary to solve a forward-backward system at each step of the iteration, as it is, for example, done in a typical mean field game (see, e.g., \cite{MFG1}, \cite{MFG3}). On the other hand, the local nature of the fixed-point problem causes additional measurability issues, in the proof of the existence result.
All these issues are addressed in Subsection \ref{subse:fixedpoint}, and the main existence result is stated in Theorem \ref{th:main}, in Section \ref{se:main}.

The literature on market microstructure is vast. Most of the theoretical work is concerned with the problem of a single agent choosing an optimal trading strategy, consisting of limit and/or market orders. The trading environment (e.g. market impact) for this agent is either specified exogenously, or is determined by the agent herself, if she is the designated market-maker.
The relevant publications include, among others, \cite{MMS2}, \cite{MMS10}, \cite{MMS11}, \cite{MMS13}, \cite{MMS20}, \cite{MMS22}, \cite{MMS23}, \cite{MMS26}, \cite{MMS25}, \cite{MMS15}, \cite{MMS15.2}, \cite{MMS28}, \cite{MMS29}, and references therein. 
Nevertheless, none of these works attempt to explain how the key market characteristics (e.g. the shape and dynamics of LOB) arise from the interaction between multiple market participants.
Finally, several recent papers have applied an equilibrium-based approach to the problem of optimal execution (cf. \cite{MMS.goe1}, \cite{MMS.goe2}). These papers describe an equilibrium between several agents solving an optimal execution problem, with the LOB (or, the market) against which these agents trade being specified exogenously, rather than being modeled as an output of the equilibrium.
The endogenous formation of LOB in an auction-style exchange (i.e. without a designated market-maker) is investigated, e.g., in \cite{MMS.g1}, \cite{MMS.g2}, \cite{MMS.g3}, \cite{MMS.g4}, \cite{MMS.g5}, \cite{MMS.g6}, \cite{DA.DuZhu}. However, the models proposed in the aforementioned papers do not aim to represent the mechanics of an auction-style exchange with sufficient precision, which is needed to address the questions we investigate herein.

The paper is organized as follows. Section \ref{se:setup} describes the proposed continuum-player game and defines the associated equilibrium.
Section \ref{se:2player} introduces an auxiliary two-player game. The latter is interesting in its own right, but its main purpose is to facilitate the construction of an equilibrium in the continuum-player game. 
The equilibria in the two-player game are described by a system of RBSDEs, whose solution components reflect against each other and whose generator does not satisfy the global Lipschitz and monotonicity properties. Proposition \ref{prop:RBSDE.exist}, in Subsection \ref{subse:RBSDE.exist}, provides the existence result for this system, which, to the best of our knowledge, has not been available before.
Section \ref{se:main} completes the construction of an equilibrium in the continuum-player game, stating the main result of the paper, Theorem \ref{th:main}. This section, in particular, describes the mollification technique for solving a fixed-point problem with discontinuity, which appears in the auxiliary continuum-player game. We believe that this method can be applied to other relevant fixed-point problems, with a similar type of discontinuity.
Finally, in Section \ref{se:example}, we consider a numerical example, in which we compute the equilibrium strategies and show how our results can be used to study the indirect market impact (illustrated by the particular case of ``spoofing").

\section{Modeling framework in continuous time}
\label{se:setup}

\subsection{Preliminary constructions}
\label{subse:preliminaries}

We consider an auction-style exchange in which the trades may occur, and the limit orders may be posted, at any time $t\in[0,T]$.
The market participants are split into two groups: the external investors, who are ``impatient", in the sense that they only submit market orders and need to execute immediately, and the strategic players, who can submit both market and limit orders, and who are willing to spend time doing so, in order to get a better execution price. In our model, we focus on the strategic players, who we refer to as \emph{agents}, and we model the behavior of the external investors exogenously, via the \emph{external demand}.
The external demand for the asset is modeled using three components: the arrival times of the \emph{potential external market orders}, the value of the \emph{potential fundamental price} at these times, and the \emph{elasticity} of the demand.
In our previous investigation \cite{GaydukNadtochiy1}, we have considered a general family of discrete time games for an auction-style exchange, with the exogenous demand process given by a discretization of a (very general) continuous time demand process, over a chosen partition of $[0,T]$. One of the main conclusions of \cite{GaydukNadtochiy1} can be, roughly, interpreted as follows: in order for a non-degenerate equilibrium\footnote{Degeneracy of an equilibrium is defined formally in \cite{GaydukNadtochiy1}. For the discussion presented herein, it suffices to know that degeneracy is an extremal state of the market, and the present work is concerned with the description of the typical (or, normal) states.} to exist in a high-frequency limit (i.e. as the diameter of the partition vanishes), the agents have to be market-neutral -- i.e. they should not expect the future fundamental price of the asset to increase or decrease. In other words, the results of \cite{GaydukNadtochiy1} seem to imply that it is hopeless to search for an equilibrium in a continuous time game (i.e. with unlimited trading frequency) in which the agents have non-trivial trading signals about the direction of the future moves of the asset price. This may sound very discouraging, however, there is a subtle feature hidden in the setting considered in \cite{GaydukNadtochiy1}. Namely, the assumptions of \cite{GaydukNadtochiy1} imply that, in the limiting high-frequency regime, the (potential) external market orders arrive with an infinite frequency, while the beliefs of the agents (i.e. their trading signals) satisfy certain continuity properties. In other words, the agents' signals are assumed to be persistent relative to the trades -- they cannot change on the same time scale on which the market orders arrive. It turns out that this assumption is crucial, and, allowing the (potential) external market orders to arrive at a finite frequency, and making the agents' beliefs be short-lived (i.e. only lasting until the next market order is executed), we can obtain a non-degenerate equilibrium in the continuous time (i.e. unlimited trading frequency) regime. Thus, herein, we model the arrival of the (potential) external market orders via a (rather general) point process, and we assume that the game ends after the first trade occurs.

Let $(\Omega,\tilde{\FF} = (\tilde{\mathcal{F}}_t)_{t\in[0,T]},\PP)$ be a stochastic basis, satisfying the usual conditions, and supporting a (multidimensional) Brownian motion $W$ and a Poisson random measure $N$. We assume that the compensator of $N$ is finite on $[0,T]\times\RR$ (i.e. $N$ is the jump measure of a compound Poisson process) and that it is absolutely continuous w.r.t. Lebesgue measure in time and space.
We denote by $\FF^W$ the usual augmented filtration generated by $W$. We assume that $W$ and $N$ are independent under $\PP$. 
The arrival times of the potential external market orders and the values of the potential fundamental price at these times are described by a counting random measure $M$ on $[0,T]\times\left(\RR\setminus\{0\}\right)$, defined as
$$
M(A)
= \int_{0}^T\int_{\RR} \bone_A\left(t,J_t(x)\right) N(dt,dx),
$$
where $J:(t,x)\mapsto J_t(x)$ is a predictable random function (as defined in \cite{JacodShiryaev}). 
We assume that $J$ is adapted to $\FF^W$ (in particular, it is independent of $N$). 
It is clear that the compensator of $M$ is finite on $[0,T]\times\RR$, it is absolutely continuous w.r.t. Lebesgue measure in time and space, and it is adapted to $\FF^W$. Then, it can be represented as $\lambda_t f_t(x)\, dt\, dx$, with an $\RR$-valued process $\lambda\geq 0$ and a random function $f:(t,x)\mapsto f_t(x)\geq 0$, progressively measurable and adapted to $\FF^W$, and s.t. $\int_{\RR} f_t(x) dx = 1$. Notice that, conditional on $\mathcal{F}^W_T$, $M$ is a Poisson random measure with the compensator $\lambda_t f_t(x)\, dt\, dx$.
The $t$-components of the atoms of $M$ are the arrival times of the potential external market orders, and their $x$-components represent the values of the potential fundamental price at these times. A positive value of $x$ corresponds to the arrival time of a potential external buy order, and a negative value corresponds to the arrival time of a potential external sell order. 
More precisely, we define the \emph{fundamental price} process $X$ (or, the \emph{reservation price} process of external investors) as the jump process of $M$:
\begin{equation}\label{eq.X.def}
X_t = \int_{\RR}xM(\{t\}\times dx).
\end{equation}
Note that the $X$ is the jump process of $M$, but it is \emph{not} a cumulative jump process: it stays at $X_0=0$ at all times except the jump times (thus, $X$ can also be interpreted as changes in the fundamental price). We choose $X_0=0$ to simplify the notation. In general, any $X_0\in\RR$ is possible, but the only effect it would have on the game is shifting all prices and values by $X_0$. To develop a better understanding of the proposed framework, from the economic point of view, it may be useful to think of $X_0$ as the \emph{last transaction price}, which occurred right before the current game started (although this is not important for the mathematical constructions).
The process $\lambda$ describes the intensity of arrival of the potential external market orders (both buy and sell). The function $f_t$ is the probability density of the value of the potential fundamental price at time $t$. We refer to $f$ as the density process of the jump sizes. 
When the jump size of the fundamental price (along with the demand elasticity, defined below) is not enough to trigger a trade, the jump remains ``unregistered" by the agents, and the fundamental price returns to zero.
The \emph{elasticity} of the external demand for the asset is described by the progressively measurable random field $D:(t,p)\mapsto D_t(p)$, adapted to $\FF^W$. We assume that, a.s., $D_t(\cdot)$ is a strictly decreasing continuous function taking value zero at zero.
Then, the total external demand to buy and sell the asset at time $t$, at the price level $p$ and at all more favorable prices, is equal to
\begin{equation}\label{eq.Dplus}
D^+_t(p) = \max\left(0, D_t(p-X_t)\bone_{\{X_t>0\}}\right),
\quad D^-_t(p) = -\min\left(0, D_t(p-X_t)\bone_{\{X_t<0\}}\right),
\end{equation}
respectively.

At any time $t$, every agent (i.e. strategic player) is allowed to submit a \emph{market order} or a \emph{limit order}. The assumptions made further in the paper make it possible to submit a limit order at such a level that it may never get executed -- this, effectively, allows the agents to wait (i.e. do nothing). We do not allow for any time-priority in the limit orders. Instead, we assume that the tick size is zero (the set of possible price levels is $\RR$), and, hence an agent can achieve a priority by posting her order slightly above or below the competing ones (and arbitrarily close to them). 
The game stops at the terminal time $T$ or at the time when the first trade occurs -- whichever one is the earliest.
The mechanics of order execution are explained in the next subsection.
There is an infinite number of agents, and the inventory of an agent is measured in ``shares per unit mass of agents" (see a discussion of this assumption in \cite{GaydukNadtochiy1}).
We assume that the agents are split into two groups: the ones whose initial inventory $s$ is positive (the \emph{long} agents, typically, indicated with a superscript ``$a$"), and those whose initial inventory $s$ is negative (the \emph{short} agents, indicated with a superscript ``$b$"). We assume that the absolute size of each agent's inventory is the same, $s\in\{-1,1\}$, and that an agent with inventory $s$ posts orders of size $s$. 
These assumptions are motivated by the results of our previous investigation \cite{GaydukNadtochiy1}, which demonstrate that, in equilibrium, the absolute value of agent's inventory only scales the size of her orders proportionally, but does not change their type and location.\footnote{Note that the precise setting and the main questions of \cite{GaydukNadtochiy1} are not the same as in the present paper. Nevertheless, the two modeling frameworks have many common features. In particular, in both cases, each agent is risk-neutral and infinitesimally small (hence, has no individual impact), which, ultimately, causes their equilibrium strategies to simply scale with the size of initial inventory.}
We also assume that we are given a pair of measurable spaces of \emph{beliefs}, $\mathbb{A}$ and $\mathbb{B}$, and, for each $\alpha\in\mathbb{A}\cup \mathbb{B}$, there exists a \emph{subjective probability measure} $\PP^{\alpha}$ on $\left(\Omega,\tilde{\FF}\right)$, which is dominated by $\PP$. An agent with beliefs $\alpha$ models the external demand under measure $\PP^{\alpha}$. The empirical distribution of the agents across beliefs is given by a pair of countably additive finite measures $\mu = (\mu^a,\mu^b)$, on $\mathbb{A}$ and $\mathbb{B}$, respectively. Note that, because the game stops right after the first market order is executed, the empirical distribution $\mu$ remains constant throughout the game. 
We make the following assumption on the measures $\{\PP^{\alpha}\}$.

\begin{ass}\label{ass:indep.underPalpha}
Under every $\PP^\alpha$, $W$ remains a Brownian motion, and the jump process of $N$ is a process with conditionally independent increments w.r.t. $\mathcal{F}_T^W$ (in the sense of \cite{JacodShiryaev}).
\end{ass}

The above assumption holds throughout the paper. 
It implies that, under every $\PP^{\alpha}$, $X$ is a process with conditionally independent increments w.r.t. $\mathcal{F}_T^W$. Using this observation and the absolute continuity of $\PP^{\alpha}$ w.r.t. $\PP$, it is easy to deduce that, under every $\PP^{\alpha}$, the compensator of the jump measure of $X$, i.e. of the measure $M$, is given by 
\begin{equation}\label{eq.lambdaF.def}
\lambda^\alpha_t f^\alpha_t(x) dt dx,
\end{equation}
with some nonnegative $\FF^W$-adapted $\lambda$ and $\FF^W$-progressively measurable $f^{\alpha}$, s.t. $\int_{\RR} f^{\alpha}_t(x) dx = 1$.
The interpretation of $\lambda^\alpha$ and $f^\alpha$ is the same as the interpretation of  $\lambda$ and $f$, but under the measure $\PP^{\alpha}$.
Note that we choose not to change the distribution of $W$ under different measures $\PP^{\alpha}$ for a technical reason -- in order to avoid $Z$-dependence in the generator of the associated RBSDE system (\ref{eq.RBSDE.Va.Vb}).

It is clear that Assumption \ref{ass:indep.underPalpha} is satisfied if $Z^{\alpha}_T=d\PP^{\alpha}/d\PP$ is given by a stochastic exponential of a process that is an integral of $\FF^W$-adapted random function w.r.t. compensated $N$. Namely,
$$
dZ^{\alpha}_t = Z^{\alpha}_{t-} \int_{\RR} \Gamma^{\alpha}_t(x)\, [N(dt,dx)-\lambda_t f_t(x) dtdx],
$$
where $\Gamma^{\alpha}\geq-1$ is $\FF^W$-progressively measurable. The compensator of $N$ under $\PP^{\alpha}$ is obtained by multiplying its compensator under $\PP$ by $1+\Gamma^{\alpha}$, hence, Assumption \ref{ass:indep.underPalpha} is clearly satisfied in this case (cf. \cite{JacodShiryaev}).
In Section \ref{se:example}, we provide an example of a family of probability measures $\{\PP^{\alpha}\}$ in the above form.

In the proposed setting, the compensator of $X$ under $\PP^{\alpha}$, given by (\ref{eq.lambdaF.def}), represents the supply/demand signal used by the agents with beliefs $\alpha$: in particular, it determines the arrival intensities of external buy and sell orders. Indeed, the value of $X$ is determined uniquely by a path of $W$ and a realization of the random measure $N$. As the compensator of $N$ may be different under each $\PP^{\alpha}$, the resulting compensator of $X$ may also vary, however, it always remains adapted to $\FF^W$. Thus, the distribution of $X$ under $\PP^{\alpha}$ is uniquely determined by the choice of $(\lambda^{\alpha},f^{\alpha})$.\footnote{To have a complete model for the external demand, one also needs to know its elasticity $D$, but the latter is $\FF$-adapted, hence, its distribution is the same under each $\PP^{\alpha}$.} As a result, the agents' beliefs can be viewed as the ``models" they use to map the observed information, given by $W$, into the predictive signal, given by (\ref{eq.lambdaF.def}).

\subsection{The continuum-player game}

Throughout the rest of this paper we, mostly, work with the filtration $\FF^W$, hence, we denote $\FF = \FF^W$.
The state of an agent is $(s,\alpha)\in \left(\{1\}\times \mathbb{A}\right)\cup \left(\{-1\}\times \mathbb{B}\right)=:\mathbb{S}$. 
Let us now discuss the controls of the agents and the order execution rules.
First, we assume that $\alpha$, representing the agent's beliefs, does not change over time.\footnote{Note that the conditional distribution of the future demand can change dynamically, according the new information revealed.}
Therefore, the state process of an agent represents only her inventory, which can only change once (because the game ends after the first trade). 
The control of every agent is given by a pair of processes $(p,v) = (p_t,v_t)_{t\in[0,T]}$, progressively measurable with respect to $\FF$.\footnote{It may seem natural to assume that the agents' filtration is enlarged by the information generated by the external trades -- i.e. by the jumps of $X$ that lead to a trade. Note that, since the game ends after the first trade, there may only be one such jump. Then, it is easy to see that the predictable filtration of the enlarged filtration, restricted to the time interval until the first trade, is $\mathbb{F}$ itself. Naturally, we require the controls to be predictable.} The process $p$ takes values in $\mathcal{P}(\RR)$, the space of probability measures on $\RR$, equipped with the weak topology, while $v$ takes values in $\RR$.
The second coordinate, $v$, determines the time at which the agent decides to submit a market order, and its formal definition is given below.
The first coordinate, $p_t$, indicates the time-$t$ distribution of the agent's limit orders across the price levels. For example, if $p_t$ is a Dirac measure located at $x$, then, at time $t$, the agent posts all her limit orders at the price level $x$.
The collection of all limit orders is described by the \emph{Limit Order Book (LOB)}, which is a pair of process $\nu=(\nu^a_t,\nu^b_t)_{t\in[0,T]}$, with values in the finite sigma-additive measures on $\RR$, adapted to $\FF$. Herein, $\nu^a_t$ corresponds to the cumulative limit sell orders, and $\nu^b_t$ corresponds to the cumulative limit buy orders, posted at time $t$.\footnote{For convenience, we sometimes refer to $\nu_t$ as a ``measure", rather than a ``pair of measures".}
The bid and ask prices at any time $t\in[0,T]$ are given by the random variables
$$
p^b_t = Q^+(\nu^b_t),
\,\,\,\,\,\,\,\,\,\,\,\,\,\,\,\,p^a_t = Q^-(\nu^a_t),
$$
respectively, where the functions $Q^-$ and $Q^+$ act on sigma-additive measures $\kappa$ on $\RR$ via
\begin{equation}\label{eq.Qpm.def}
Q^+(\kappa) = \sup \text{supp}(\kappa),
\,\,\,\,\,\,\,\,\,\,\,\,\,\,\,\,Q^-(\kappa) = \inf \text{supp}(\kappa).
\end{equation}
Notice that $p^b_t$ and $p^a_t$ are always well defined as extended random variables, but may take infinite values.

Assume that, at time $t$, an agent posts a limit sell order at the price level $p'$.
If the demand to buy the asset at or below the price level $p'$, $D^+_{t}(p')$, exceeds the amount of all limit sell orders posted below $p'$ at time $t$, i.e. $D^+_{t}(p') > \nu^a_t((-\infty,p'))$, then the limit sell order of the agent is executed. Analogous execution rules hold for the limit buy orders.
Thus, if an agent follows the limit order strategy $p$, her limit order is (partially) executed by an external market order at the time
$$
T^{p,a} = \inf\{t\in[0,T]\,:\, D^+_{t}\left(Q^-(p_t)\right) > \nu^a_t\left((-\infty,Q^-(p_t))\right)\},
$$
$$
T^{p,b} = \inf\{t\in[0,T]\,:\, D^-_{t}\left(Q^+(p_t)\right) > \nu^b_t\left((Q^+(p_t),\infty)\right)\},
$$
for the long and short agents, respectively.
Let us clarify the meaning of the above formulas. Assume, for simplicity (and only for the sake of this example), that the demand elasticity curve, $D$, is deterministic.
Note that $D^+_t\equiv D^-_t\equiv0$ unless $X$ jumps at $t$. Thus, the above formulas say that a non-zero fraction of agent's limit orders is executed at time $t$, by an external order, if and only if $X$ jumps at time $t$, and its jump is sufficiently large, so that the demand at the agent's ``best limit order" is higher than the size of all limit orders with higher price priority. The latter, along with continuity of $D$, ensures that a non-zero fraction of agent's limit orders is executed at this time.

The value of $v_t$ indicates the critical level of the bid or ask price (i.e. a threshold), at which the agent decides to submit a market order. We assume that the size of the agent's market order is equal to her inventory, and it is executed at the bid or ask price available at the time when the order is submitted.
Thus, the agent will submit her own market order at the time
$$
\tau^{v,a} = \inf\{t\in[0,T]\,:\, v_t \leq p^b_t\},
\,\,\,\,\,\,\,\,\,\,\,\tau^{v,b} = \inf\{t\in[0,T]\,:\, v_t \geq p^a_t\},
$$
for the long and short agents, respectively.\footnote{It is clear that, for every stopping time $\tau^{v,a/b}$ with respect to $\FF$, there exists a process $v_t$, adapted to $\FF$, such that $\tau^{v,a/b}$ has the above representation.}
The collection of all thresholds $v$ is described by the pair of processes $\theta=(\theta^a_t,\theta^b_t)_{t\in[0,T]}$, with values in the finite sigma-additive measures on $\RR$, adapted to $\FF$. 

\begin{remark}
The above definitions of the execution times make use of the assumption that each agent is infinitesimally small, and, hence, her order is necessarily executed once the demand reaches it.
They also use the following two implicit assumptions: each agent believes that her limit order will be executed first among all orders at the same price level, and her market order will be executed at the best price available. These assumptions and their connection to a finite-player game are discussed in \cite{GaydukNadtochiy1}.
\end{remark}

Recall that each agent is infinitesimal, hence, even if she executes a non-zero fraction of her inventory, this may not constitute a trade of non-zero size. We, therefore, define the first ``significant" execution time as the first time when a non-zero mass of agents execute a non-zero fraction of their inventory (i.e. when a non-zero total inventory mass is traded).
Consider the first significant execution times of external market orders:
\begin{equation}\label{eq.Ta.def}
T^a = \inf\{t\in[0,T]\,:\, D^+_{t}(p^a_t) > 0\},
\,\,\,\,\,\,
T^b = \inf\{t\in[0,T]\,:\, D^-_{t}(p^b_t) > 0\},
\end{equation}
Similarly, we define the first significant execution times of internal market orders:
\begin{equation}\label{eq.taua.def}
\tau^a = \inf\{t\in[0,T]\,:\, \theta^a_t ((-\infty,p^b_t])>0\},
\,\,\,\,\,\,\,\,\,\,\,\,
\tau^b = \inf\{t\in[0,T]\,:\, \theta^b_t ([p^a_t,\infty))>0\}.
\end{equation}
Finally, given $(\nu,X,D)$, we define the \emph{clearing prices}:
$$
\tilde{p}^{c,a}_t = \sup\{p< Q^+(\nu^a_t)\,:\,D^+_{t}(p)>\nu^a_{t}((-\infty,p))\},
\,\,\,\,\,\,\,
p^{c,a}_t = \tilde{p}^{c,a}_t \bone_{\{\tilde{p}^{c,a}_t \geq p^{a}_t\}},
$$
$$
\tilde{p}^{c,b}_t = \inf\{p> Q^-(\nu^b_t)\,:\,D^-_{t}(p)>\nu^b_{t}((p,\infty))\},
\,\,\,\,\,\,\,
p^{c,b}_t = \tilde{p}^{c,b}_t \bone_{\{\tilde{p}^{c,b}_t \leq p^{b}_t\}}.
$$
For a long agent with strategy $(p,v)$, the game ends at the time $T^{p,a}\wedge \tau^{v,a}\wedge T\wedge T^a\wedge T^b\wedge \tau^a\wedge \tau^b$ (and similarly for the short agents). If an agent has any inventory left at the end of the game, then it is marked to market.\footnote{There is no canonical way to choose the marking-to-market rules in a setting where agents have no exogenously given valuation of the asset (and we insist on using such a setting, because we think of the agents as ``pure speculators"). In particular, other marking rules are possible. Herein, we merely make a choice of marking rules which is economically meaningful.}
The precise rules for computing the payoff of a long agent, using strategy $(p,v)$, are described below.

{\bf If the game is terminated by an external market order}: $T^{p,a}\wedge T^a\wedge T^b < T\wedge\tau^a\wedge\tau^b$ (note that equality is impossible, as the right hand side is predictable and the left hand side is totally inaccessible).
\begin{itemize}
\item If $T^{p,a}\wedge T^a < T^b$ (equality is impossible), then the payoff is 
\begin{equation}\label{eq.payoff.complicated}
\int_{-\infty}^{\tilde{p}^{c,a}_t} z p_t(dz) + \int_{\tilde{p}^{c,a}_t}^{\infty} (p^{c,a}_{t} + p^b_{t}) p_t(dz),
\,\,\,\,\,\,\,\,\,\text{with}\,\,\,\,\,t=T^{p,a}\wedge T^a.
\end{equation}
\item If $T^b < T^{p,a}\wedge T^a$, then the payoff is $p^b_{T^b} + p^{c,b}_{T^b}$.
\end{itemize}
Notice that the remaining inventory of an agent is marked to the bid price shifted by the clearing price. This choice can be (heuristically) interpreted as follows.
Assume that, after the trade, a new game starts, with the agents having the same distribution of inventory and the same beliefs about the distribution of future jumps of $X$ (i.e. the same $\{(\lambda^{\alpha},f^{\alpha})\}$). Then, the only parameter that is different in the new game, as compared to the original one, is the value of $X_0$, which, in the new game, becomes equal to the clearing price. As mentioned in the discussion following (\ref{eq.X.def}), the new value of $X_0$ will simply shift all prices and values in the new game by $X_0$, hence, the bid price is shifted by the value of clearing price. Finally, it is easy to deduce (and will be shown later in the paper) that it is suboptimal for an agent to post a limit buy order at positive levels. Thus, if an external sell order is executed, the clearing price is non-positive, and, hence, the remaining inventory is marked to the current bid price shifted downwards (the opposite holds if an external buy order is executed).

{\bf If the game is terminated by an internal market order}: $T\wedge\tau^a\wedge\tau^b < T^{p,a}\wedge T^a\wedge T^b$.
\begin{itemize}
\item If $\tau^b < \tau^a \wedge T$ then the payoff is $p^a_{\tau^b}$.
\item If $\tau^a\wedge T \leq \tau^b$ then the payoff is $p^b_{\tau^a\wedge T}$.
\end{itemize}
To explain the above, assume, e.g., that an internal buy order occurs: i.e. $\tau^b < \tau^a \wedge T$.
Note that the internal orders are different, because they are predictable. Hence, the long agents can act exactly at the time $\tau^b$ and ``flock" their limit orders to the best ask price, $p^a$, to match the market orders from short agents (who initiated the internal buy order). On the other hand, if any of the agents (long or short) do not trade at $p^a$, they will mark their inventory to the bid or ask price shifted by $p^a$, and, since $p^b\leq0\leq p^a$, it is easy to see that it is beneficial for all of them to trade at $p^a$.\footnote{Of course, in practice, not all agents will act at the same time: only a fraction of them will submit the internal market orders at the end of the game, the others will move on to the next game, with updated $X_0$. However, such ``flocking" of agents at the end of the game (provided the game ends with an internal market order) is consistent with the empirical observation of ``clustering trades".}

The following diagram (containing a reference to equation (\ref{eq.payoff.complicated})) describes the payoff of a \emph{long} agent:

\begin{tikzpicture}
  [
    grow                    = right,
    sibling distance        = 4em,
    level distance          = 20em,
    edge from parent/.style = {draw, -latex},
    every node/.style       = {font=\footnotesize},
    sloped
  ]
  \node [root] {$t=0$}
      child { node [dummy] {}
      child{node[env]{$p^b_{\tau^a}$}
      edge from parent node [below] {internal sell}}
      child{node[env]{$p^a_{\tau^b}$}
      edge from parent node [above] {internal buy}}
      edge from parent node [below] {internal market order} 
      }
      child { node [env] {$p^b_{T}$}
      edge from parent node [above] {no market orders} 
      }
    child { node [dummy] {}
      child{node[env]{$p^b_{T^b} + p^{c,b}_{T^b}$}
      edge from parent node [below] {external sell}}
      child{node[env]{(\ref{eq.payoff.complicated})}
      edge from parent node [above] {external buy}}
      edge from parent node [above] {external market order} 
      };
\end{tikzpicture}

Similar rules apply to short agents.
Formally, given $(\nu,\theta,X,D)$, the individual objective of an agent starting at the initial state $(1,\alpha)$ and using the control $(p,v)$ is given by:
\begin{equation}\label{eq.intro.Jlong.def}
J^{(\nu,\theta),(p,v)}(1,\alpha) =
\EE^{\alpha} \left[\int_{\RR} \left(z \bone_{\{ z\leq \tilde{p}^{c,a}_{\hat{T}^{p,a}}\}} 
+ \left(p^b_{\hat{T}^{p,a}} + p^{c,a}_{\hat{T}^{p,a}} \right) \bone_{\{z>\tilde{p}^{c,a}_{\hat{T}^{p,a}}\}}\right)p_{\hat{T}^{p,a}}(dz) 
\bone_{\{\hat{T}^{p,a} < T^b\wedge \hat{\tau}^{v,a}\wedge\tau^b\}} 
\right.
\end{equation}
$$
\left.
+\left( p^b_{T^b} + p^{c,b}_{T^b} \right) \bone_{\left\{T^b < \hat{T}^{p,a} \wedge \hat{\tau}^{v,a}\wedge\tau^b\right\}} 
+ \left( p^a_{\tau^b}\bone_{\{\tau^b<\hat{\tau}^{v,a}\}} + p^b_{\hat{\tau}^{v,a}}\bone_{\{\tau^b\geq\hat{\tau}^{v,a}\}} \right)
\bone_{\{\hat{T}^{p,a}\wedge T^b > \hat{\tau}^{v,a}\wedge\tau^{b}\}}
\right]
$$
where $\hat{T}^{p,a} = T\wedge T^{p,a}\wedge T^a$, $\hat{\tau}^{v,a}=T\wedge\tau^{v,a}\wedge\tau^a$, and we assume that $0\cdot\infty = 0$.
Similarly,
\begin{equation}\label{eq.intro.Jshort.def}
J^{(\nu,\theta),(p,v)}(-1,\alpha) =
\EE^{\alpha} \left[-\int_{\RR} \left(z \bone_{\{z\geq \tilde{p}^{c,b}_{\hat{T}^{p,b}}\}} 
+ \left(p^a_{\hat{T}^{p,b}} + p^{c,a}_{\hat{T}^{p,b}} \right) \bone_{\left\{z<\tilde{p}^{c,b}_{\hat{T}^{p,b}}\right\}}
\right)p_{\hat{T}^{p,b}}(dz) 
\bone_{\{\hat{T}^{p,b} < T^a\wedge \hat{\tau}^{v,b}\wedge\tau^a\}} 
\right.
\end{equation}
$$
\left.
-\left( p^a_{T^b} + p^{c,a}_{T^a} \right) \bone_{\{T^a < \hat{T}^{p,b} \wedge \hat{\tau}^{v,b}\wedge\tau^a\}} 
- \left( p^b_{\tau^a}\bone_{\{\tau^a<\hat{\tau}^{v,b}\}} + p^a_{\hat{\tau}^{v,b}}\bone_{\{\tau^a\geq\hat{\tau}^{v,b}\}} \right)
\bone_{\{\hat{T}^{p,b}\wedge T^a > \hat{\tau}^{v,b}\wedge\tau^{a}\}}
\right]
$$
where $\hat{T}^{p,b} = T\wedge T^{p,b}\wedge T^b$, $\hat{\tau}^{v,b}=T\wedge\tau^{v,b}\wedge\tau^b$.
Every agents aims to maximize her objective. The above objectives may seem convoluted -- this is because they are meant to provide a close approximation of the real-world execution rules and marking to market. In the next subsection, we establish a more transparent representation of the objectives.

In the following definitions, we assume that a stochastic basis, a Brownian motion $W$, a random measure $M$, a random field $D$, spaces $\mathbb{A}$ and $\mathbb{B}$, an associated set of measures $\{\PP^{\alpha}\}_{\alpha\in \mathbb{A}\cup \mathbb{B}}$, and the empirical distribution $\mu$, are fixed and satisfy the assumptions made earlier in this section. (Nevertheless, it is shown in Subsection \ref{subse:representation} that the input $(M,\{\PP^{\alpha}\})$ can be replaced by the agents' signals $\{\lambda^{\alpha},f^{\alpha}\}$.)

\begin{definition}\label{def:admis}
For a given market $(\nu,\theta)$ and a state $(s,\alpha)\in \mathbb{S}$, a pair of $\FF$-progressively measurable processes $(p,v)$ is an {\bf admissible} control, if the positive part of the expression inside the expectation in (\ref{eq.intro.Jlong.def}) (if $s=1$) or (\ref{eq.intro.Jshort.def}) (if $s=-1$) has a finite expectation under $\PP^{\alpha}$.
\end{definition}

\begin{definition}\label{def:opt}
For a given market $(\nu,\theta)$ and state $(s,\alpha)\in \mathbb{S}$, we call an admissible control $(p,v)$ {\bf optimal} if
$$
J^{(\nu,\theta),(p,v)}(s,\alpha) \geq J^{(\nu,\theta),(p',v')}(s,\alpha)
$$ 
$\PP$-a.s., for any admissible control $(p',v')$.
\end{definition}

In the above, we make the standard assumption of games with \emph{a continuum of players}: each agent is too small to affect the distribution of cumulative controls (described by $\nu$) when she changes her control.
Next, we define Nash equilibrium in the proposed game.

\begin{definition}\label{def:equil.def}
A given market $(\nu,\theta)$ and a pair of $\FF$-progressively measurable random fields $(p,v):\Omega\times[0,T]\times\mathbb{S}\rightarrow\mathcal{P}(\RR)\times\RR$ form an {\bf equilibrium}, if
\begin{enumerate}

\item for $\mu$-a.e. $(s,\alpha)\in\mathbb{S}$, $(p(s,\alpha),v(s,\alpha))$ is an optimal control for $(\nu,\theta)$ and $(s,\alpha)$,

\item and the following holds $\PP$-a.s., for any $t<\bar{T}:=T\wedge T^a\wedge T^b \wedge\tau^a\wedge\tau^b)$ and any $x\in\RR$:
\begin{equation}\label{eq.nuplus.fixedpoint.def}
\nu^a_t((-\infty,x]) = \int_{\mathbb{A}} p_t\left(1,\alpha;\left(-\infty,x\right]\right) \mu^a(d\alpha),
\,\,\,\,\,\,
\nu^b_t((-\infty,x]) = \int_{\mathbb{B}} p_t\left(-1,\alpha;\left(-\infty,x\right]\right) \mu^b(d\alpha),
\end{equation}
\begin{equation}\label{eq.numinus.fixedpoint.def}
\theta^a_t((-\infty,x]) = \int_{\mathbb{A}} \bone_{\left\{v_t(1,\alpha)\leq x \right\}} \mu^a(d\alpha),
\,\,\,\,\,\,
\theta^b_t((-\infty,x]) = \int_{\mathbb{B}} \bone_{\left\{v_t(-1,\alpha)\leq x \right\}} \mu^b(d\alpha).
\end{equation} 
\end{enumerate}
\end{definition}
Note that a trivial equilibrium, in which all agents stop immediately, is always possible. However, such equilibrium, clearly, is not sufficient for modeling purposes, and the existence of other, non-trivial, equilibria is far from obvious. In the remainder of this paper, we use an auxiliary two-player game (cf. Section \ref{se:2player}) to identify a class of more realistic potential equilibria, in which the end time of the game is determined uniquely by the solution of an associated RBSDE system (cf. (\ref{eq.Y1.Y2.RBSDE})), and we prove the existence of equilibrium in this class, in Theorem \ref{th:main}. Even though it is possible to construct models in which the resulting equilibrium is still trivial (i.e. the end time of the game is zero), this is not the case in general, as confirmed by the example in Section \ref{se:example}. 

\begin{remark}\label{re:inconsist}
In the above definition, it is implicitly assumed that the empirical measure of the agents' states remains constant in time until the game is over for all players. This is, indeed, the case, if the equilibrium is such that, $\PP$-a.s., for all $t<\bar{T}$, we have:
\begin{equation}\label{eq.endog.mu.def}
\mu\circ\left( (s,\alpha)\mapsto S_t(s,\alpha) \right)^{-1} = \mu,
\end{equation}
with 
$$
S_t(1,\alpha) = \bone_{\left[0,T^{p(1,\alpha),a}\wedge \tau^{v(1,\alpha),a}\right)}(t),\,\,\, \text{and}\,\,\, S_t(-1,\alpha) = - \bone_{\left[0,T^{p(-1,\alpha),b}\wedge \tau^{v(-1,\alpha),b}\right)}(t).
$$
The condition (\ref{eq.endog.mu.def}) may fail if a non-zero mass of agents manages to execute their orders strictly before $\bar{T}$: i.e. if $T^{p(1,\alpha),a}\wedge \tau^{v(1,\alpha),a}<\bar{T}$ for a set of $\alpha$ with positive $\mu^a$-measure, or $T^{p(-1,\alpha),b}\wedge \tau^{v(-1,\alpha),b}<\bar{T}$ for a set of $\alpha$ with positive $\mu^b$-measure. The latter cannot occur due to external market orders, because they only arrive at a finite number of times and, before $T^a\wedge T^b\geq \bar{T}$, only a zero mass of agents can execute their limit orders against any such market order (cf. (\ref{eq.Ta.def})).
It is also true that, at any time $t$, before $\tau^a\wedge\tau^b\geq\bar{T}$, only a zero mass of agents can execute their internal market orders (cf. (\ref{eq.taua.def})). However, the set of such times $t$ may be uncountable.
Therefore, to ensure that $\mu$ remains constant and, hence, (\ref{eq.endog.mu.def}) holds, it suffices to consider only the equilibria satisfying, $\PP$-a.s., for all $t$, except, possibly, a countable set:
$$
v_t(1,\alpha)\geq v^a_t:= Q^-(\theta^a_t),
\,\,\,\,\,\,\,\,\,\,\,\,\,\,\,v_t(-1,\alpha)\leq v^b_t:= Q^+(\theta^b_t),
\,\,\,\,\,\,\,\,\,\,\,\,\,\,
\forall \alpha\in\mathbb{A}\cup\mathbb{B}.
$$
In the subsequent sections, we construct such an equilibrium.
\end{remark}

\subsection{Representation of the objective}
\label{subse:representation}

In this section, we provide an equivalent representation of the objective of the agents, which makes it more tractable and more convenient for the analysis that follows.
In addition, it shows that the main input parameters for the proposed equilibrium problem are the signals $\{(\lambda^{\alpha},f^{\alpha})\}_{\alpha\in\mathbb{A}\cup\mathbb{B}}$, forming the compensators of $X$ under $\{\PP^{\alpha}\}$, and the demand elasticity $D$ (the latter is independent of $\alpha$ and, in many realistic models, can be deterministic). In particular, there is no need to keep track of the random measure $N$ and the probability measures $\{\PP^{\alpha}\}$ -- they are only needed to show that the present setting fits within the standard framework for games with heterogenous beliefs.
The desired representation is derived following standard arguments, making use of the independence of the driving Poisson measure $N$ and the Brownian motion $W$.
First, we introduce new notation that will be used throughout the paper. For any $\alpha\in\mathbb{A}\cup\mathbb{B}$, $t\in[0,T]$, $p,x,y\in\RR$ and $\kappa\in\mathcal{P}(\RR)$, we define the \emph{instanteneous filling rates} for limit orders at levels $x$ and $y$:
\begin{equation}\label{eq.F.c.def}
F_t^{+,\alpha}(x) = \int_{x\vee0}^{\infty} f^\alpha_t(u)\text{d}u,
\,\,\,\,\,\,\,\,
F_t^{-,\alpha}(y) = \int_{-\infty}^{y\wedge0} f^\alpha_t(u)\text{d}u,
\,\,\,\,\,\,\,\,
c_t^\alpha(x,y) = \lambda_t^\alpha \left( F^{-,\alpha}_t(y) + F^{+,\alpha}_t(x) \right).
\end{equation}
Next, we define the \emph{clearing price} as a function of the fundamental price $x$:
\begin{equation}\label{eq.lca.def}
l^{c,a}_t(x) = \sup\left\{p<Q^+(\nu^a_t)\,:\,D_t(p-x)>\nu^a_t((-\infty,p))\right\},
\end{equation}
\begin{equation}\label{eq.lcb.def}
l^{c,b}_t(x) = \inf\left\{p>Q^-(\nu^b_t)\,:\,-D_t(p-x)>\nu^b_t((p,\infty))\right\}.
\end{equation}
Notice that, if $X$ has a positive jump at time $t$, then the clearing price at time $t$ is given by $\tilde{p}^{c,a}_t = l^{c,a}_t(X_t)$. Similarly, if $X$ has a negative jump at time $t$, then $\tilde{p}^{c,b}_t = l^{c,b}_t(X_t)$.
Finally, we introduce the \emph{instantaneous reward rates} from executed limit orders, distributed according to $\kappa$, with the bid and ask prices $y$ and $x$: 
\begin{equation}\label{eq.ha.def}
h^{\alpha,a}_t(\kappa,x,y) = \lambda^{\alpha}_t \int_{(Q^-(\kappa)\wedge x)\vee0}^{\infty} f^\alpha_t(u) 
\left[
\int_{-\infty}^{l^{c,a}_t(u)} z\kappa(dz)
+ \left(y + l^{c,a}_t(u)\bone_{\{l^{c,a}_t(u)\geq x\}}\right) \kappa\left((l^{c,a}_t(u),\infty)\right)
\right]du
\end{equation}
$$
+ \lambda^{\alpha}_t \int_{-\infty}^{y\wedge0} f^\alpha_t(u) \left(y + l^{c,b}_t(u)\right) du,
$$
\begin{equation}\label{eq.hb.def}
h^{\alpha,b}_t(\kappa,x,y) = \lambda^{\alpha}_t \int_{-\infty}^{(Q^+(\kappa)\vee y)\wedge0} f^\alpha_t(u) 
\left[
\int_{l^{c,b}_t(u)}^{\infty} z\kappa(dz)
+ \left(x + l^{c,b}_t(u)\bone_{\{l^{c,b}_t(u)\leq y\}}\right) \kappa\left((-\infty,l^{c,b}_t(u))\right)
\right] du
\end{equation}
$$
+ \lambda^{\alpha}_t \int_{x\vee0}^{\infty} f^\alpha_t(u) \left(x + l^{c,a}_t(u)\right) du.
$$
Using the above notation, we can obtain a simplified expression for the objective, given in the following lemma. Note that the expectation in this representation is taken under the reference measure, and the objective depends only on the cumulative actions $(\nu,\theta)$ and on $(\{\lambda^{\alpha},f^{\alpha}\},D)$ (as the expressions in (\ref{eq.F.c.def})--(\ref{eq.hb.def}) depend only on $(\{\lambda^{\alpha},f^{\alpha}\},D)$). 

\begin{lemma}\label{prop:simpObj}
Let Assumption \ref{ass:indep.underPalpha} hold.
Given a market $(\nu,\theta)$, for any $\alpha\in\mathbb{A}\cup\mathbb{B}$ and any admissible strategy $(p,v)$, we have:
\begin{equation}\label{eq.Jrep.long}
J^{(\nu,\theta),(p,v)}(1,\alpha) = 
\EE\Big[\int_0^{\hat{\tau}^{v,a}\wedge\tau^b} \exp\left(-\int_0^s c^{\alpha}_u\left(p^a_u\wedge Q^-(p_u),p^b_u\right)du\right) h^{\alpha,a}_s(p_s,p^a_s,p^b_s)ds
\end{equation}
\begin{equation*}
+ \exp\left(-\int_0^{\hat{\tau}^{v,a}\wedge\tau^b} c^{\alpha}_u\left(p^a_u\wedge Q^-(p_u),p^b_u\right) du\right)
\left(p^a_{\tau^b}\bone_{\{\tau^b<\hat{\tau}^{v,a}\}} + p^b_{\hat{\tau}^{v,a}} \bone_{\{\tau^b\geq\hat{\tau}^{v,a}\}}\right)\Big],
\end{equation*}
\begin{equation}\label{eq.Jrep.short}
J^{(\nu,\theta),(p,v)}(-1,\alpha) = 
-\EE\Big[\int_0^{\hat{\tau}^{v,b}\wedge\tau^a} \exp\left(-\int_0^s c^{\alpha}_u\left(p^a_u,p^b_u\vee Q^+(p_u)\right)du\right) h^{\alpha,b}_s(p_s,p^a_s,p^b_s)ds
\end{equation}
\begin{equation*}
+ \exp\left(-\int_0^{\hat{\tau}^{v,a}\wedge\tau^b} c^{\alpha}_u\left(p^a_u,p^b_u\vee Q^+(p_u)\right) du\right)
\left(p^b_{\tau^a}\bone_{\{\tau^a<\hat{\tau}^{v,b}\}} + p^a_{\hat{\tau}^{v,b}} \bone_{\{\tau^a\geq\hat{\tau}^{v,b}\}}\right)\Big],
\end{equation*}
where $\hat{\tau}^{v,a}=T\wedge \tau^{v,a}\wedge\tau^a$, $\hat{\tau}^{v,b}=T\wedge \tau^{v,b}\wedge\tau^b$ and the expectations are taken under $\PP$.
\end{lemma}
\begin{proof}
The proof follows easily by conditioning on $W$. Notice that, conditional on $\mathcal{F}_T$, $M$ is a Poisson random measure, with the deterministic compensator $\lambda^{\alpha}_t f^{\alpha}_t(x)\,dt\,dx$, which is finite on $[0,T]\times\RR$. Recall also that $D$, $\nu$, $\theta$, $p$, $v$, $p^a$, $p^b$, $\tau^{v,a}$, $\tau^{v,b}$, $\tau^a$, $\tau^b$, and all the random functions defined above the lemma, are adapted to $\FF$. Conditional on $\mathcal{F}_T$, they become deterministic functions of time. 
Recall the fundamental price process, $X_t = \int_{\RR}x M(\{t\}\times dx)$, and introduce
$$
Y_t = X_t \left(\bone_{\{X_t>(p^a_t\wedge Q^-(p_t))\vee0\}} + \bone_{\{X_t<p^b_t\wedge0\}}\right).
$$
Notice that $\hat{T}^{p,a}$ is the time of the first positive jump of $Y_t$, and $T^b$ is the time of its first negative jump. 
Notice also that, conditional on $\mathcal{F}_T$, the clearing price $\tilde{p}^{c,a}_t$ becomes a deterministic function of $t$ and $Y_t$: $\tilde{p}^{c,a}_t = l^{c,a}_t(Y_t)$.
Thus, conditional on $\mathcal{F}_T$, the expression inside the expectation in (\ref{eq.intro.Jlong.def}) becomes a function of the time and size of the first jump of $Y$. 
Conditional on $\mathcal{F}_T$, $X$ is the jump process of a Poisson random measure with the compensator $\lambda^{\alpha}_t f^{\alpha}_t(u)dudt$. 
It is also clear that, conditional on $\mathcal{F}_T$, $Y$ is the jump process of a non-homogeneous compound Poisson process with intensity $c^{\alpha}_t\left(p^a_t\wedge Q^-(p_t),p^b_t\right)$, and with the distribution of jump sizes at time $t$ given by
$$
\frac{\lambda^{\alpha}_t f^{\alpha}(x)}{c^{\alpha}_t\left(p^a_t\wedge Q^-(p_t),p^b_t\right)}
\left(\bone_{\{x\leq p^b_t\wedge0\}} + \bone_{\{x\geq (p^a_t\wedge Q^-(p_t))\vee0\}} \right)dx.
$$
A standard computation, then, yields (\ref{eq.Jrep.long}). The equation (\ref{eq.Jrep.short}) is derived similarly. The expectations in (\ref{eq.Jrep.long}) and (\ref{eq.Jrep.short}) are taken under $\PP$, because the expressions inside the expectations are adapted to $\FF=\FF^W$, and $W$ has the same distribution under $\PP$ and $\PP^{\alpha}$.
\qed
\end{proof}

\section{A two-player game}
\label{se:2player}

In this section, we consider an auxiliary non-zero-sum two-player control-stopping game. It is related to the continuum-player game, but the precise connection will be established in the subsequent sections.
We refer the reader to \cite{KaratzasLi}, \cite{KaratzasSuderth}, and the references therein, for more on non-zero-sum two-player control-stopping games.\footnote{See, e.g., \cite{Dyn1}, \cite{Dyn2}, \cite{Dyn3}, \cite{Dyn4}, and the references therein, for the related classical Dynkin games, which are zero-sum and stopping-only.} It is worth mentioning, however, that the present game does not fall within any of the classes considered before. A more detailed description of this class of games is carried out in our forthcoming work \cite{GaydukNadtochiy3}.

Assume that all the probabilistic constructions made in Subsection \ref{subse:preliminaries} are in place. Namely, we are given a stochastic basis, with a Brownian motion $W$, a Poisson measure $N$, a counting random measure $M$, a family of probability measures $\{\PP^{\alpha}\}$, and with the demand elasticity process $D$, as described in Section \ref{se:setup}. We assume that Assumption \ref{ass:indep.underPalpha} holds. Assume, in addition, that $\mathbb{A}=\{\alpha^0\}$ and $\mathbb{B}=\{\beta^0\}$. Consider a two-player game, in which the first (long) player starts with the initial inventory $1$ and has beliefs $\alpha^0$, and the second (short) player starts with the initial inventory $-1$ and has beliefs $\beta^0$. The game proceeds according to the rules similar to those described in the previous section: each agent can post limit orders on the respective side of the book, or can terminate the game by submitting a market order. The execution of limit orders against the external market orders occurs in exactly the same way as described in the previous section. However, herein, at any given time, each agent is only allowed to post limit orders at a single location (i.e. the control $p_t$ is a Dirac measure). In addition, the main difference between the present game and the one defined in the previous section is that, herein, each player has a non-zero mass and, hence, can affect the LOB. In fact, since there is only one player on each side of the book, the LOB is given by a combination of two Dirac measures: $\nu^a_t = \delta_{p^a_t}$, $\nu^b_t = \delta_{p^b_t}$, controlled by the locations of the players' limit orders: $p^a$ for the long agent, and $p^b$ for the short one. Clearly, $p^a$ also coincides with the ask price, and $p^b$ is the bid price. Note that each of these prices is now controlled by a single agent, which is not the case in the original game described in the previous section. The same is true for the stopping thresholds: $\theta^a$ and $\theta^b$ are given by Dirac measures, and the locations of these measures correspond to the thresholds $v^a$ and $v^b$ used by the long and short agents, respectively. In this new game (due to its simplicity), it turns out to be more convenient to work with the associated stopping times $\tau^a$ and $\tau^b$. In fact, we will further constraint the agents' controls, so that $\tau^a=\tau^b=:\tau$ and $p^a_{\tau}=p^b_{\tau}=\bar{p}_{\tau}$. The meaning behind these constraints is clear: every agent assumes that the counterparty will execute a market order at exactly the same time as she does, and that these orders are executed at the same price. 
Taking into account the above considerations, we transform (\ref{eq.intro.Jlong.def}) into the objective of a long player:
\begin{equation}\label{eq.intro.tildeJlong.def}
\tilde{J}^{a,(p^b,\bar{p}),(p,\tau)} =
\EE^{\alpha^0} \left[ p_{T^{p,a}} \bone_{\{T^{p,a} < T^b \wedge \tau\}} 
+ 2p^b_{T^b} \bone_{\left\{T^b < T^{p,a} \wedge \tau\right\}} 
+ \bar{p}_{\tau} \bone_{\{T^{p,a}\wedge T^b > \tau\}}
\right],
\end{equation}
where $p^b$, $\bar{p}$ and $p$ are $\RR$-valued $\FF$-adapted processes, $\tau$ is a stopping time with values in $[0,T]$, and
$$
T^{b} = \inf\{t\in[0,T]\,:\,X_t < p^b_t\},
\,\,\,\,\,\,\,T^{p,a} = \inf\{t\in[0,T]\,:\,X_t > p_t\},
\,\,\,\,\,\,\,X_t = M(\{t\}\times\RR).
$$
Similarly, for the short agents,
\begin{equation}\label{eq.intro.tildeJshort.def}
\tilde{J}^{b,(p^a,\bar{p}),(p,\tau)} =
-\EE^{\alpha^0} \left[ p_{T^{p,b}} \bone_{\{T^{p,b} < T^a \wedge \tau\}} 
+ 2p^a_{T^a} \bone_{\left\{T^a < T^{p,b} \wedge \tau\right\}} 
+ \bar{p}_{\tau} \bone_{\{T^{p,b}\wedge T^a > \tau\}}
\right],
\end{equation}
where $p^a$, $\bar{p}$ and $p$ are $\RR$-valued $\FF$-adapted processes, $\tau$ is a stopping time with values in $[0,T]$, and
$$
T^{a} = \inf\{t\in[0,T]\,:\,X_t > p^a_t\},
\,\,\,\,\,\,\,\,\,T^{p,b} = \inf\{t\in[0,T]\,:\,X_t < p_t\}.
$$
Using Lemma \ref{prop:simpObj}, we deduce the following form of the objective functions
\begin{equation}\label{eq.intro.tildeJlong.def.simp}
\tilde{J}^{a,(p^b,\bar{p}),(p,\tau)} =
\EE\Big[
\int_0^{\tau} \exp\left(-\int_0^s c^{\alpha^0}_u(p_u,p^b_u) du\right)g^{a}_s(p_s,p^b_s)ds
\end{equation}
\begin{equation*}
+\exp\left(-\int_0^{\tau} c^{\alpha^0}_u(p_u,p^b_u)du\right)\bar{p}_{\tau}
\Big],
\end{equation*}
where $c^{\alpha}$ is defined in (\ref{eq.F.c.def}) and
\begin{equation}\label{eq.ga.def}
g_t^{a}(x,y) = \lambda_t^{\alpha^0} \left( 2yF^{\alpha^0,-}_t(y) + xF^{\alpha^0,+}_t(x) \right).
\end{equation}
Similarly,
\begin{equation}\label{eq.intro.tildeJshort.def.simp}
\tilde{J}^{b,(p^a,\bar{p}),(p,\tau)} =
-\EE\Big[
\int_0^{\tau} \exp\left(-\int_0^s c^{\beta^0}_u(p^a_u,p_u) du\right)g^{b}_s(p^a_s,p_s)ds
\end{equation}
\begin{equation*}
+\exp\left(-\int_0^{\tau} c^{\beta^0}_u(p^a_u,p_u)du\right) \bar{p}_{\tau}
\Big],
\end{equation*}
where
\begin{equation}\label{eq.gb.def}
g_t^{b}(x,y) = \lambda_t^{\beta^0} \left( yF^{\beta^0,-}_t(y) + 2xF^{\beta^0,+}_t(x) \right).
\end{equation}
To ensure that the above expressions are well defined, and to analyze the equilibrium in a two-player game, we need to make the following assumptions.

\begin{ass}\label{ass:bdd.lambda.f}
There exists a constant $C'>0$, s.t., $\PP$-a.s., $|\lambda^\alpha_t|, |f^{\alpha}_t(x)|\leq C'$, for all $\alpha\in\mathbb{A}\cup\mathbb{B}$, $t\in[0,T]$ and $x\in\RR$.
\end{ass}
We also assume that the possible price jump sizes are bounded.
\begin{ass}
There exists a constant $C_p>0$, s.t., $\PP$-a.s., $\supp (f^\alpha_t)\subseteq [-C_p,C_p]$, for all $\alpha\in\mathbb{A}\cup\mathbb{B}$ and $t\in[0,T]$.
\end{ass}

Denote by $\mathbb{S}^2$ the set of continuous $\FF$-adapted processes $Y$, such that $\sup_{0\le t\le T}|Y_t|\in\mathbb{L}^2$. We say that the terminal execution price $\bar{p}$ is admissible if $\bar{p}\in\mathbb{S}^2$. A control $(p,\tau)$ is admissible if $p$ is $\FF$-progressively measurable, satisfying, $\PP$-a.s., $|p_t|\leq C_p$ for all $t\in[0,T]$, and $\tau$ is $\FF$-stopping time.
Next, we introduce the notions of optimality and equilibrium in the two-player game -- they are analogous to Definitions \ref{def:opt}--\ref{def:equil.def}.

\begin{definition}\label{def:opt.2player}
For a given admissible $(p^b,\bar{p})$, we call an admissible control $(p,\tau)$ \emph{optimal for the long agent} if
$$
\tilde{J}^{a,(p^b,\bar{p}),(p,\tau)} \geq \tilde{J}^{a,(p^b,\bar{p}),(p',\tau')},
$$ 
for any admissible control $(p',\tau')$. Similarly, for a given admissible $(p^a,\bar{p})$, we call an admissible control $(p,\tau)$ \emph{optimal for the short agent} if
$$
\tilde{J}^{b,(p^a,\bar{p}),(p,\tau)} \geq \tilde{J}^{b,(p^a,\bar{p}),(p',\tau')},
$$ 
for any admissible control $(p',\tau')$.
\end{definition}

\begin{definition}\label{def:equil.def.2player}
A combination $(p^a,p^b,\tau,\bar{p})$ is an \emph{equilibrium in the two-player game}, if it is admissible and, given $(p^b,\bar{p})$, the control $(p^a,\tau)$ is optimal for the long agent, while, given $(p^a,\bar{p})$, the control $(p^b,\tau)$ is optimal for the short agent.
\end{definition}
In the next subsection, we characterize the equilibrium strategies via a system of Reflected Backward Stochastic Differential Equations (RBSDEs).

\subsection{Characterizing the equilibria via a system of RBSDEs}

The next assumptions are used to guarantee the uniqueness and regularity of the optimal control of an agent.

\begin{ass}\label{ass:f.cont}
$\PP$-a.s., for any $\alpha\in\mathbb{A}\cup\mathbb{B}$ and $t\in[0,T]$,
$f^\alpha_t(\cdot)$ is continuous in the interior of its support, $f^\alpha_t(0)=0$, and $0<F^{+,\alpha}_t(0)<1$.
\end{ass}

\begin{ass}\label{ass:f.logconcave}
$\PP$-a.s., for any $\alpha\in\mathbb{A}$ and $t\in[0,T]$,
$F^{+,\alpha}_t(\cdot)/f^{\alpha}_t(\cdot)$ is a decreasing function in the interior of $\supp (f^{\alpha}_t)\cap\mathbb{R}_+$, vanishing at the right end of the interval. 
Similarly, $\PP$-a.s., for any $\beta\in\mathbb{B}$ and $t\in[0,T]$, $F^{-,\beta}_t(\cdot)/f^{\beta}_t(\cdot)$ is an increasing function in the interior of $\supp (f^{\beta}_t)\cap\mathbb{R}_-$, vanishing at the left end of the interval.
\end{ass}

\begin{remark}
The monotonicity of $F^{+,\alpha}_t(\cdot)/f^{\alpha}_t(\cdot)$, for example, is implied by the log-concavity of the distribution of positive jumps (similarly, for the negative jumps).
Instead of requiring that $F^{+,\alpha}_t(\cdot)/f^{\alpha}_t(\cdot)$ is decreasing, it suffices to assume that its growth rate is bounded from above by $1-\varepsilon$, for a constant $\varepsilon>0$ independent of $(t,\omega)$.
\end{remark}


To prove the existence of a solution to a system of RBSDEs characterizing the equilibria in a two-player game, we also need to assume that ``the range of beliefs is relatively bounded".

\begin{ass}\label{ass:bdd.range}
There exists a constant $C>0$, s.t., $\PP$-a.s.:
$$
\frac{1}{C}\le \left\vert\frac{\lambda^{\alpha^0}_t}{\lambda^{\beta^0}_t}\right\vert \le C,\quad\frac{1}{C}\le \left\vert\frac{f^{\alpha^0}_t(x)}{f^{\beta^0}_t(x)}\right\vert \le C,
\quad \forall\,x\in\RR\,\,t\in[0,T].
$$
\end{ass}

First we analyze the individual optimization problem of an agent, taking the actions of the counterparty as given. Assume that we are given a process $\bar{p}\in\mathbb{S}^2$ and progressively measurable $(p^a,p^b)$, such that $\PP$-a.s., $|p^a_t|,|p^b_t|\le C_p$, $\forall t\in[0,T]$.
Let us introduce the value functions of the agents:
\begin{equation}\label{eq.intro.tildeVlong.def.simp}
V^a_t = \operatorname{ess\,sup}\limits_{\tau\in\mathcal{T}_{t},\,p} \EE\Big[
\int_t^{\tau} \exp\left(-\int_t^s c^{\alpha^0}_u(p_u,p^b_u) du\right)g^{a}_s(p_s,p^b_s)ds
\end{equation}
\begin{equation*}
+\exp\left(-\int_t^{\tau} c^{\alpha^0}_u(p_u,p^b_u)du\right) \bar{p}_{\tau}
\Big\vert \mathcal{F}_t \Big],
\end{equation*}
\begin{equation}\label{eq.intro.tildeVshort.def.simp}
V^b_t = \operatorname{ess\,inf}\limits_{\tau\in\mathcal{T}_{t},\,p} \EE\Big[
\int_t^{\tau} \exp\left(-\int_t^s c^{\beta^0}_u(p^a_u,p_u) du\right)g^{b}_s(p^a_s,p_s)ds
\end{equation}
\begin{equation*}
+\exp\left(-\int_t^{\tau} c^{\beta^0}_u(p^a_u,p_u)du\right)\bar{p}_{\tau}
\Big\vert \mathcal{F}_t\Big],
\end{equation*}
where $\mathcal{T}_{t}$ is the set of $\FF$-stopping times with values in $[t,T]$, $p$ is any $\FF$-progressively measurable process, with $|p|\leq C_p$, and $c^{\alpha}$, $g^{a}$ and $g^{b}$ are defined in (\ref{eq.F.c.def}), (\ref{eq.ga.def}) and (\ref{eq.gb.def}).
In addition, we introduce the following random functions:
$$
\mathcal{G}^{a,x}_t(y,z)=-c^{\alpha^0}_t(x,z)y + g^{a}_t(x,z),
\,\,\,\,\,\,\,\,x,y\in\RR,
$$
$$
\mathcal{G}^a_t(y,z) = \sup_{x\in\RR}\mathcal{G}^{a,x}_t(y,z) = -c^{\alpha^0}_t\left(P^a_t(y),z\right)y + g^{a}_t\left(P^a_t(y),z\right),\,\,\,\,\,\,\,\,y\in\RR,
$$
where $P^a_t$ provides the optimal price location at the ask side, given in a feedback form:
\begin{equation}\label{eq.Pa.def}
P_t^a(y) = \inf \arg\max_{p\in\RR} (p-y) F_t^{+,\alpha^0}(p),\,\,\,\,\,\,\,\,y\in\RR.
\end{equation}
Similarly, for any admissible $p^a$, we define
\begin{equation}\label{eq.Pb.def}
P_t^b(y) = \sup\arg\max_{p\in\RR} (y-p)F_t^{-,\beta^0}(p),\,\,\,\,\,\,\,\,y\in\RR,
\end{equation}
$$
\mathcal{G}^b_t(z,y) = -c^{\beta^0}_t\left(z,P^b_t(y)\right)y + g^{b}_t\left(z,P^b_t(y)\right),\,\,\,\,\,\,\,\,y\in\RR.
$$
The value of $P^a_t(y)$ can be described as the unique nonnegative solution $p$ of 
\begin{equation}\label{eq.Pa.root.def}
p-y=F^{+,\alpha^0}_t(p)/f^{\alpha^0}_t(p),
\end{equation} 
unless $y$ is too large, in which case $P^a_t(y)$ is the upper boundary of the support of $f^{\alpha^0}_t$, or too small, in which case $P^a_t(y)=0$. Similarly, $P^b_t(y)$ is the unique non-positive solution $p$ of 
\begin{equation}\label{eq.Pb.root.def}
y-p=F^{-,\beta^0}_t(p)/f^{\beta^0}_t(p),
\end{equation}
or the lower boundary of the support of $f^{\beta^0}_t$, if $y$ is too small, or zero, if $y$ is too large.

\begin{lemma}\label{le:Pa.Pb.prop}
Let Assumptions \ref{ass:indep.underPalpha}--\ref{ass:f.logconcave} hold.
Then, the random functions $P^a$ and $P^b$ are progressively measurable and satisfy, $\PP$-a.s., for all $t\in[0,T]$: 
$$
0\leq P^a_t(y)\le C_p,\,\,\,\,\,-C_p\le P^b_t(y)\leq 0,\,\,\,\,\,
P^a_t(y)\ge y,\,\,\,\,\,P^b_t(y)\le y,\,\,\,\,\,\forall y\in\RR,
$$
and, in addition, $P^a_t(\cdot)$ and $P^b_t(\cdot)$ are non-decreasing and 1-Lipschitz.
\end{lemma}
\begin{proof}
The progressive measurability property and the above inequalities follow directly from Assumptions \ref{ass:bdd.lambda.f}--\ref{ass:f.cont}. The monotonicity and 1-Lipschitz property follow from Assumption \ref{ass:f.logconcave} and the representations (\ref{eq.Pa.root.def})--(\ref{eq.Pb.root.def}).
\qed
\end{proof}

The above lemma, along with Assumptions \ref{ass:bdd.lambda.f}--\ref{ass:f.cont}, implies that, for any admissible $(p,p^b,\bar{p})$, $\mathcal{G}^a_t(0,p^b_t)$ and $\mathcal{G}^{a,p_t}_t(0,p^b_t)$ are bounded processes, and that $\mathcal{G}^a_t(y,p^b_t)$ and $\mathcal{G}^{a,p_t}_t(y,p^b_t)$ are Lipschitz in $y$, uniformly over a.e. $(t,\omega)$.
This allows us to use Proposition 7.1 from \cite{ELK}, to show that, for any admissible $(p,p^b,\bar{p})$, the process $Y$, which is a continuous modification of
$$
Y_t:=\hat{J}^{a,(p^b,\bar{p}),p}_t=\operatorname{ess\,sup}\limits_{\tau\in\mathcal{T}_{t}} \EE\Big[\int_t^{\tau} \exp\left(-\int_t^s c^{\alpha^0}_u(p_u,p^b_u) du\right) g^{a}_s(p_s,p^b_s)ds
$$
$$
+\exp\left(-\int_t^\tau c^{\alpha^0}_u(p_u,p^b_u)du\right)\bar{p}_\tau\Big\vert \mathcal{F}_t\Big],
$$
is the unique $\mathbb{S}^2$ solution of the affine RBSDE,
\begin{eqnarray}
-dY_t=\mathcal{G}^{a,p_t}_t(Y_t,p^b_t)dt - Z_t\text{d}W_t + \text{d}K_t\quad0\le t\le T
\label{eq.BSDE.2player.affine.first}\\
Y_t\ge \bar{p}_t \quad 0\le t \le T,\quad \int_0^T (Y_t-\bar{p}_t)\text{d}K_t=0 \\
Y_T=\bar{p}_T,
\label{eq.BSDE.2player.affine.last}
\end{eqnarray}
where $Z$ is a progressively measurable square-integrable (multidimensional) process, $K\in\mathbb{S}^2$ is increasing and satisfies $K_0=0$.
Similarly, the existence results from \cite{ELK} imply that 
\begin{eqnarray}
-dY_t = \mathcal{G}^a_t(Y_t,p^b_t)dt - Z_t\text{d}W_t + \text{d}K_t\quad0\le t\le T
\label{eq.BSDE.2player.long.first}\\
Y_t\ge \bar{p}_t \quad 0\le t \le T\quad \int_0^T (Y_t-\bar{p}_t)\text{d}K_t=0 \\
Y_T=\bar{p}_T
\label{eq.BSDE.2player.long.last}
\end{eqnarray}
has a unique solution $(Y_t,Z_t,K_t)$.
Then, Theorem 7.2 in \cite{ELK} implies that $Y$ is a continuous modification of $V^a$, and that $p^a_t=P^a_t(Y_t)$ and $\tau^a=\inf\{s\ge 0\colon Y_s=\bar{p}_s\}$ form an optimal control for the long agent.
Similarly, for a given admissible $(p^a,\bar{p})$, there exists a unique solution $(Y_t,Z_t,K_t)$ to
\begin{eqnarray}
-dY_t = \mathcal{G}^b_t(p^a_t,Y_t)dt - Z_t\text{d}W_t - \text{d}K_t\quad0\le t\le T
\label{eq.BSDE.2player.short.first}\\
Y_t\le\bar{p}_t \quad 0\le t \le T,\quad \int_0^T (\bar{p}_t-Y_t)\text{d}K_t=0 \\
Y_T=\bar{p}_T,
\label{eq.BSDE.2player.short.last}
\end{eqnarray}
$Y$ is a continuous modification of $V^b$, and $p^b_t=P^b_t(Y_t)$ and $\tau^b=\inf\{s\ge 0\colon Y_s=\bar{p}_s\}$ form an optimal control for the short agent. 
It turns out that, because the optimal stopping time has to be the same for both agents in equilibrium, we can formulate a system of equations for $V^a$ and $V^b$ without $\bar{p}$.
In order to state this result formally, we need to introduce the following random functions
\begin{equation}\label{eq.Ga.def}
\tilde{\mathcal{G}}^a_t(y,z) = \mathcal{G}^a_t(y,P^b_t(z)) = -c^{\alpha^0}_t\left(P^a_t(y),P^b_t(z)\right)y + g^{a}_t\left(P^a_t(y),P^b_t(z)\right),\,\,\,\,\,\,\,y,z\in\RR,
\end{equation}
\begin{equation}\label{eq.Gb.def}
\tilde{\mathcal{G}}^b_t(y,z) = \mathcal{G}^b_t(P^a_t(y),z) = -c^{\beta^0}_t\left(P^a_t(y),P^b_t(z)\right)z + g^{b}_t\left(P^a_t(y),P^b_t(z)\right),\,\,\,\,\,\,\,y,z\in\RR,
\end{equation}
where $c^{\alpha}$, $g^{a}$ and $g^{b}$ are defined, respectively, in (\ref{eq.F.c.def}), (\ref{eq.ga.def}) and (\ref{eq.gb.def}), and $P^a$ and $P^b$ are given by (\ref{eq.Pa.def}) and (\ref{eq.Pb.def}).

\begin{lemma}\label{sab}
Let Assumptions \ref{ass:indep.underPalpha}--\ref{ass:bdd.range} hold.
For any equilibrium $(p^a,p^b,\tau,\bar{p})$ in the two-player game (in the sense of Definition \ref{def:equil.def.2player}), the value functions of the agents, $V^a,V^b\in\mathbb{S}^2$, satisfy 
\begin{equation}\label{eq.RBSDE.Va.Vb}
\left\{
\begin{array}{l}
-\text{d}V^a_t = \tilde{\mathcal{G}}^a_t(V^a_t,V^b_t)dt - Z^a_t dW_t + dK^a_t  \phantom{\frac{1}{\frac{1}{2}}}\\
-\text{d}V^b_t = \tilde{\mathcal{G}}^b_t(V^a_t,V^b_t)dt - Z^b_tdW_t - dK^b_t  \phantom{\frac{1}{\frac{1}{2}}}\\
V^a_t\ge V^b_t\quad\forall t\in[0,T], \quad
\int_0^T (V^a_t - V^b_t) d(K^a_t+K^b_t)=0 \phantom{\frac{1}{\frac{1}{2}}}\\
V^a_T = V^b_T,  \phantom{\frac{1}{\frac{1}{2}}}
\end{array}
\right.
\end{equation}
with some increasing processes $K^a,K^b\in\mathbb{S}^2$, starting at zero, and with progressively measurable square-integrable $(Z^a,Z^b)$. Moreover, $(\hat{p}^a,\hat{p}^b,\hat{\tau},\bar{p})$ also form an equilibrium, with the same value functions, where: $\hat{p}^a_t=P^a_t(V^a_t)$, $\hat{p}^b_t=P^b_t(V^b_t)$ and $\hat{\tau}=\inf\{s\ge 0\colon V^a_s = V^b_s\}$. 
Conversely, given a solution to (\ref{eq.RBSDE.Va.Vb}), we can define the optimal controls $(\hat{p}^a,\hat{p}^b,\hat{\tau})$ as above, and choose $\bar{p}=(1-\eta)V^a+\eta V^b$, with any progressively measurable process $\eta$ taking values in $(0,1)$, to obtain an equilibrium $(\hat{p}^a,\hat{p}^b,\hat{\tau},\bar{p})$.
\end{lemma}
\begin{proof}
Consider an equilibrium $(p^a,p^b,\tau,\bar{p})$. As discussed earlier, the standard results on BSDEs (cf. \cite{ELK}) imply that $(V^a,Z^a,K^a)$ solves (\ref{eq.BSDE.2player.long.first})--(\ref{eq.BSDE.2player.long.last}), and $(V^b,Z^b,K^b)$ solves (\ref{eq.BSDE.2player.short.first})--(\ref{eq.BSDE.2player.short.last}) (both systems are considered with the same $\bar{p}$). 
It follows from the optimality of $\tau$, via the standard theory, that $V^b_{\tau}=\bar{p}_{\tau}=V^a_{\tau}$. 
Consider the long agent. It is clear that the objective of the long agent cannot increase if we replace $\bar{p}$ by $V^b$ in its definition (cf. (\ref{eq.intro.tildeJlong.def.simp})). On the other hand, $\tau$ is optimal and $\bar{p}_{\tau}=V^b_{\tau}$, hence, the value function $V^a$ remains the same if we replace $\bar{p}$ by $V^b$ in its definition (cf. (\ref{eq.intro.tildeVlong.def.simp})). Therefore,  $(V^a,Z^a,K^a)$ solves  (\ref{eq.BSDE.2player.long.first})--(\ref{eq.BSDE.2player.long.last}) with $\bar{p}$ replaced by $V^b$.
Similar argument applies to the short agent, and yields that $(V^b,Z^b,K^b)$ solves  (\ref{eq.BSDE.2player.short.first})--(\ref{eq.BSDE.2player.short.last}) with $\bar{p}$ replaced by $V^a$.
Next, using the optimality of $p^a$ and the comparison principle for the BSDE (\ref{eq.BSDE.2player.affine.first}), we easily deduce that, for a.e. $(t,\omega)$, $p^a_t$ coincides with $\hat{p}^a_t=P^a_t(V^a_t)$ whenever $\lambda^{\alpha^0}_t>0$ and $V^a_t<\sup \text{supp} (f^{\alpha^0}_t)$. On the other hand, Assumption \ref{ass:bdd.range} implies that, if $\lambda^{\alpha^0}_t=0$ or $V^a_t\geq\sup \text{supp}( f^{\alpha^0}_t)$, then $\lambda^{\beta^0}_t=0$ or $V^a_t\geq\sup \text{supp} f^{\beta^0}_t$, and, in turn, $\mathcal{G}^b_t\left(p^a_t,V^b_t\right) = \mathcal{G}^b_t\left(\hat{p}^a_t,V^b_t\right)$. Thus, we conclude that $V^b$ satisfies (\ref{eq.BSDE.2player.short.first})--(\ref{eq.BSDE.2player.short.last}) with $p^a$ replaced by $\hat{p}^a$. Similarly, we conclude that $V^a$ satisfies (\ref{eq.BSDE.2player.long.first})--(\ref{eq.BSDE.2player.long.last}) with $p^b$ replaced by $\hat{p}^b$.
Thus, $(V^a,V^b)$ satisfy (\ref{eq.RBSDE.Va.Vb}).

Next, consider a solution to (\ref{eq.RBSDE.Va.Vb}). Choosing $\bar{p}$ as shown in the statement of the lemma, we conclude that $(V^a,Z^a,K^a)$ solves (\ref{eq.BSDE.2player.long.first})--(\ref{eq.BSDE.2player.long.last}), with $p^b$ replaced by $\hat{p}^b$. Then, the standard results (cf. \cite{ELK}) imply that, given $\hat{p}^b$ and $\bar{p}$, $V^a$ is the value function of the long agent, and her optimal control is given by $\hat{p}^a$ and
$$
\inf\{s\ge 0\colon V^a_s\leq\bar{p}_s\}=\inf\{s\ge 0\colon V^a_s = V^b_s\}=\hat{\tau}.
$$ 
Similar argument applies to the short agent, completing the proof.
\qed
\end{proof}

\subsection{Existence of a solution}
\label{subse:RBSDE.exist}

In this subsection, we address the question of existence of a solution to the RBSDE (\ref{eq.RBSDE.Va.Vb}). 
The main difficulty in analyzing (\ref{eq.RBSDE.Va.Vb}) is the non-standard form of reflection: the components of the solution reflect against each other, as opposed to reflecting against a given boundary.
Related equations have been analyzed in the literature on BSDEs arising in the problem of optimal switching: see, e.g., \cite{oblique}, \cite{oblique2}, and the references therein. However, the exact form of reflection in (\ref{eq.RBSDE.Va.Vb}) is different, and its generator does not possess the desired monotonicity properties, making it impossible to prove the existence of a solution to (\ref{eq.RBSDE.Va.Vb}) using the methods developed in optimal switching literature.
Before we analyze the existence, it is convenient to consider the question of uniqueness. Note that there are two reflecting components of the solution, but only one minimality constraint, which indicates the potential lack of uniqueness of a solution to (\ref{eq.RBSDE.Va.Vb}).
The possibility of an arbitrary choice of $\eta$ in Lemma \ref{sab} leads to the same conclusion. Indeed, a different choice of $\eta$ produces a different $\bar{p}$, which results in a different pair of value functions $(V^a,V^b)$, which, nevertheless, have to solve the same system (\ref{eq.RBSDE.Va.Vb}). This heuristic observation turns out to be correct and, in fact, allows us to construct a solution to (\ref{eq.RBSDE.Va.Vb}).
Consider a solution $(V^a,V^b,K^a,K^b,Z^a,Z^b)$ to (\ref{eq.RBSDE.Va.Vb}). 
Introducing $K_t=K^a_t+K^b_t$, we notice that there must exist a process $\eta$, with values in $[0,1]$, such that $\text{d}K^a_t=\eta_t\text{d}K_t$, $\text{d}K^b_t=(1-\eta_t)\text{d}K_t$. 
Then, we introduce the new variables $(\tilde{Y}^1,\tilde{Y}^2)$, s.t. $\tilde{Y}^1_t = V^a_t - V^b_t$ and $d\tilde{Y}^2_t=(1-\eta_t)dV^a_t + \eta_t dV^b_t$, to replace $(V^a,V^b)$. 
Assuming that the change of variables can be inverted, one obtains a system of RBSDEs for $(\tilde{Y}^1,\tilde{Y}^2)$, in which only the first component reflects against zero, and $\tilde{Y}^1_T=0$. 
Conversely, we can start by prescribing $\eta$ and a terminal condition for $\tilde{Y}^2$, solving the associated system of RBSDEs for $(\tilde{Y}^1,\tilde{Y}^2)$, and, then, recover $(V^a,V^b)$ from $(\tilde{Y}^1,\tilde{Y}^2,\eta)$ via the above formulas. Naturally, the resulting $(V^a,V^b)$ are expected to satisfy (\ref{eq.RBSDE.Va.Vb}).
This method seems to describe all solutions to (\ref{eq.RBSDE.Va.Vb}), however, herein, we are only interested in constructing a particular one.\footnote{It is an interesting topic for future research, to describe rigorously all solutions of (\ref{eq.RBSDE.Va.Vb}).} 
Hence, we choose $\eta\equiv1/2$ and $\tilde{Y}^2_T=0$, to obtain $Y^1=\tilde{Y}^1=V^a - V^b$ and $Y^2=2\tilde{Y}^2=V^a + V^b$, which are expected to satisfy:
\begin{equation}\label{eq.Y1.Y2.RBSDE}
\begin{cases}
-\text{d}Y^1_t = \mathcal{G}^1_t(Y^1_t,Y^2_t)\text{d}t-Z^1_t\text{d}W_t+\text{d}K_t\phantom{\frac{1}{\frac{1}{2}}} \\
Y^1_t\ge0,\,\,\,\,\,\,\,\,\, \int_0^T Y^1_t\text{d}K_t=0,\quad Y^1_T=0\phantom{\frac{1}{\frac{1}{2}}}\\
-\text{d}Y^2_t = \mathcal{G}^2_t(Y^1_t,Y^2_t)\text{d}t-Z^2_t\text{d}W_t,\quad Y^2_T=0 \phantom{\frac{1}{\frac{1}{2}}}
\end{cases}
\end{equation}
where $Y^1,Y^2\in\mathbb{S}^2$, the processes $Z^1,Z^2$ are progressively measurable and square-integrable, $K\in\mathbb{S}^2$ is increasing and satisfies $K_0=0$. In addition, we denote

\begin{equation*}
\mathcal{G}^1_t(y^1,y^2) = \tilde{\mathcal{G}}^a_t\left((y^1+y^2)/2,(y^2-y^1)/2\right) 
- \tilde{\mathcal{G}}^b_t\left((y^1+y^2)/2,(y^2-y^1)/2\right),
\end{equation*}
\begin{equation*}
\mathcal{G}^2_t(y^1,y^2) = \tilde{\mathcal{G}}^a_t\left((y^1+y^2)/2,(y^2-y^1)/2\right) 
+ \tilde{\mathcal{G}}^b_t\left((y^1+y^2)/2,(y^2-y^1)/2\right)
\end{equation*}
where $\tilde{\mathcal{G}}^a$ and $\tilde{\mathcal{G}}^b$ are defined in (\ref{eq.Ga.def}) and (\ref{eq.Gb.def}).
The following lemma formalizes the connection between (\ref{eq.Y1.Y2.RBSDE}) and (\ref{eq.RBSDE.Va.Vb}), and its proof follows easily by a direct verification.

\begin{lemma}\label{le:connect.RBSDEs}
Let $(Y^1,Y^2,Z^1,Z^2,K)$ be a solution to (\ref{eq.Y1.Y2.RBSDE}). Then 
$$
V^a=\frac{1}{2}Y^1+\frac{1}{2}Y^2,\,\, V^b=\frac{1}{2}Y^2-\frac{1}{2}Y^1,\,\, 
Z^a=\frac{1}{2}Z^1+\frac{1}{2}Z^2,\,\, Z^b = \frac{1}{2}Z^2-\frac{1}{2}Z^1,\,\,
K^a=\frac{1}{2}K, K^b=\frac{1}{2}K
$$
form a solution to (\ref{eq.RBSDE.Va.Vb}).
\end{lemma}

Note that the specific choice of $\eta\equiv1/2$ corresponds to choosing an angle of reflection of the process $(V^a,V^b)$ against the straight line ``$V^a=V^b$" in $\RR^2$. The specific angle chosen to obtain (\ref{eq.Y1.Y2.RBSDE}) implies orthogonal reflection against this line, and (\ref{eq.Y1.Y2.RBSDE}) arises after a simple rotation, which turns this line into a horizontal axis. The systems of RBSDEs with orthogonal reflection in a general convex domain have been analyzed in \cite{OrthRBSDE}. However, the latter results are not applicable in the present case, as the generator of (\ref{eq.Y1.Y2.RBSDE}) lacks the global Lipschitz property.
Indeed, the generator can be written as
\begin{equation}\label{eq.G1.def}
\mathcal{G}^1_t(y^1,y^2) = -c^1_t(y^1,y^2) y^1 + c^2_t(y^1,y^2) y^2 + g^1_t(y^1,y^2),
\end{equation}
\begin{equation}\label{eq.G2.def}
\mathcal{G}^2_t(y^1,y^2) = -c^2_t(y^1,y^2) y^1 - c^1_t(y^1,y^2) y^2 + g^2_t(y^1,y^2),
\end{equation}
where
$$
c^1_t(y^1,y^2) = \frac{1}{2} c^{\alpha^0}_t\left(P^a_t\left((y^1+y^2)/2\right),P^b_t\left((y^2-y^1)/2\right)\right) 
+ \frac{1}{2} c^{\beta^0}_t\left(P^a_t\left((y^1+y^2)/2\right),P^b_t\left((y^2-y^1)/2\right)\right)
$$
$$
c^2_t(y^1,y^2) = \frac{1}{2} c^{\beta^0}_t\left(P^a_t\left((y^1+y^2)/2\right),P^b_t\left((y^2-y^1)/2\right)\right)
- \frac{1}{2} c^{\alpha^0}_t\left(P^a_t\left((y^1+y^2)/2\right),P^b_t\left((y^2-y^1)/2\right)\right),
$$
$$
g^1_t(y^1,y^2) = g_t^{a}\left(P^a_t\left((y^1+y^2)/2\right),P^b\left((y^2-y^1)/2\right)\right) 
- g_t^{b}\left(P^a\left((y^1+y^2)/2\right),P^b\left((y^2-y^1)/2\right)\right),
$$
$$
g^2_t(y^1,y^2) = g_t^{a}\left(P^a_t\left((y^1+y^2)/2\right),P^b\left((y^2-y^1)/2\right)\right) 
+ g_t^{b}\left(P^a\left((y^1+y^2)/2\right),P^b\left((y^2-y^1)/2\right)\right),
$$
with $c^{\alpha}$, $P^a$, $P^b$, $g^{a}$ and $g^{b}$ defined in (\ref{eq.F.c.def}), (\ref{eq.Pa.def}), (\ref{eq.Pb.def}), (\ref{eq.ga.def}) and (\ref{eq.gb.def}). 
It is easy to see that every $c^i_t(\cdot,\cdot)$ and $g^i_t(\cdot,\cdot)$ is bounded and globally Lipschitz, uniformly over a.e. $(t,\omega)$. However, due to the presence of the multipliers $y^1$ and $y^2$, $\mathcal{G}^i_t(\cdot,\cdot)$ is unbounded and does not possess the global Lipschitz property.
In addition, the existence and uniqueness result established below (cf. Proposition \ref{prop:RBSDE.exist}) holds for any choice of constant $\eta$ in $(0,1)$, which, in turn, implies the ``oblique" (i.e. non-orthogonal) reflection of $(V^a,V^b)$ against the boundary, and brings the resulting system outside the scope of \cite{OrthRBSDE}.

Recall that the existence result for BSDEs with linear growth, and without global Lipschitz property, has only been established in a one-dimensional case, whereas the present equation is multidimensional. 
Nevertheless, we can make use of the fact that the generator of (\ref{eq.Y1.Y2.RBSDE}) has the ``correct" asymptotic behavior, to prove the existence of a solution. In particular, we exploit the fact that, due to the assumptions made earlier in this section, whenever $\|(Y^1_t,Y^2_t)\|$ becomes large, the generator $(\mathcal{G}^1_t,\mathcal{G}^2_t)$ pushes $(Y^1_t,Y^2_t)$ in the direction in which the largest $|Y^i_t|$ decreases.

\begin{proposition}\label{prop:RBSDE.exist}
Let Assumptions \ref{ass:bdd.lambda.f}--\ref{ass:bdd.range} hold.
Then, there exists a solution to (\ref{eq.Y1.Y2.RBSDE}), s.t. its components $Y^1$ and $Y^2$ are absolutely bounded by a constant. Such a solution is unique.
\end{proposition}

\begin{proof}
\noindent\emph{Step 1: Existence for the fully capped system.} 
For any constant $C>0$, denote $\Psi_C(y)=(-C\vee y)\wedge C$. Clearly, this function is $1$-Lipschitz in $y$ and absolutely bounded by $C$. 
We fix arbitrary constants $\{C_i^j>0\}$ and consider the fully capped system:
\begin{equation}\label{capped}
\left\{
\begin{array}{l}
-\text{d}Y^1_t = \left(-c^1_t(Y^1_t,Y^2_t) \Psi_{C_1^1}(Y^1_t)
+ c^2_t(Y^1_t,Y^2_t) \Psi_{C_1^2}(Y^2_t) 
+ g_t^1(Y^1_t,Y^2_t) \right) \text{d}t - Z^1_t\text{d}W_t + \text{d}K_t \\
-\text{d}Y^2_t = \left(-c^2_t(Y^1_t,Y^2_t) \Psi_{C_2^1}(Y^1_t)
- c^1_t(Y^1_t,Y^2_t) \Psi_{C_2^2}(Y^2_t)
+ g_t^2(Y^1_t,Y^2_t)\right)\text{d}t - Z^2_t\text{d}W_t 
\end{array}
\right.
\end{equation}
Here, and in some expressions that follow, we omit the terminal condition, barrier, and the minimality condition for $K_t$, as they remain unchanged throughout. 
Assumptions \ref{ass:bdd.lambda.f}--\ref{ass:bdd.range} imply that $c^1_t(y^1,y^2)$, $c^2_t(y^1,y^2)$, $g^1_t(y^1,y^2)$ and $g^2_t(y^1,y^2)$ are bounded and globally Lipschitz in $(y^1,y^2)$, uniformly over a.e. $(t,\omega)$. Hence, the generator of (\ref{capped}) is globally Lipschitz in $(y^1,y^2)$ (and independent of $(Z^1,Z^2)$), and the standard existence results for Lipschitz BSDEs (cf. for example, Theorem 2.2 in \cite{zhenhua}) yield the existence (and uniqueness) of a solution to (\ref{capped}). 
Denote the $Y$-component of this solution $(Y^{1c}_t,Y^{2c}_t)$.

\noindent\emph{Step 2: Bounds on solution components via partial uncapping.} We want to bound the components $(Y^{1c}_t$, $Y^{2c}_t)$, of the solution to the capped system, by using the control-stopping interpretation of the individual (R)BSDEs comprising our system. 
Consider the associated equation for $Y^1$, with $Y^{2c}_t$ being given:
\begin{equation}\label{1c}
\left\{
\begin{array}{l}
-\text{d}Y^1_t = \left(-c^1_t\left(Y^1_t,Y^{2c}_t\right) Y^1_t 
+ c^2_t\left(Y^1_t,Y^{2c}_t\right) \Psi_{C_1^2}\left(Y^{2c}_t\right) 
+ g_t^1\left(Y^1_t,Y^{2c}_t\right) \right)\text{d}t
- Z^1_t\text{d}W_t + \text{d}K_t \\
Y^1_t\ge0,\,\,\,\,\,\,\,\int_0^T Y^1_t\text{d}K_t=0,\quad Y^1_T=0 
\end{array}
\right.
\end{equation}
Note that, as $c^1_t$,$c^2_t$, $g^1$ and $\Psi_{C_1^2}$ are bounded, this one-dimensional RBSDE has a continuous generator with linear growth in $Y^1$, and, for example, by Theorem 4.1 in \cite{zhenhua}, it has a solution, which we denote $Y^1_t$.
Next, for $Y^1$ and $Y^{2c}$ constructed above, we introduce the processes 
$$
\tilde{c}^1_t = c^1_t(Y^1_t,Y^{2c}_t),\quad\tilde{c}^2_t = c^2_t(Y^1_t,Y^{2c}_t),
\quad
\tilde{g}^1_t = g^1_t(Y^1_t,Y^{2c}_t),\quad\tilde{g}^2_t = g^2_t(Y^1_t,Y^{2c}_t),
$$
and consider the one-dimensional RBSDE (for $\tilde{Y}$), obtained from (\ref{1c}) by pretending that the coefficients should depend on the solution itself:
\begin{equation}
\left\{
\begin{array}{l}
-\text{d}\tilde{Y}^1_t = \left(-\tilde{c}^1_t \tilde{Y}^1_t + \tilde{c}^2_t \Psi_{C_1^2}(Y^{2c}_t)
+\tilde{g}_t^1\right)\text{d}t - Z^1_t\text{d}W_t + \text{d}K_t \\
\tilde{Y}^1_t\ge0,\,\,\,\,\,\,\,\int_0^T \tilde{Y}^1_t\text{d}K_t=0,\quad \tilde{Y}^1_T=0
\end{array}
\right.
\end{equation}
Note that $\tilde{Y}=Y^1$ is the unique solution of this equation. On the other hand, the above RBSDE is affine in $\tilde{Y}$, and, for example, by Theorem 7.1 in \cite{ELK}, its unique solution admits the following interpretation, as the value function of an optimal stopping problem:
$$
Y^1_t=\sup\limits_{\tau\in\mathcal{T}_{t}}\EE\left[
\int_t^{\tau} \exp\left(-\int_t^s \tilde{c}^1_u \text{d}u\right) 
\left( \tilde{c}^2_s \Psi_{C_1^2}(Y^{2c}_s) + \tilde{g}_s^1\right) \text{d}s
\big\vert \mathcal{F}_t\right]
$$
We will use this representation to establish a bound on $|Y^1|$. 
First, note that, under our assumptions, there exist constants $C_0>0$ and $\lambda\in(0,1)$, such that, for all $t$, $y^1$, $y^2$, and a.e. $\omega$, we have:
$$
\left\vert\frac{g^i_t(y^1,y^2)}{c^1_t(y^1,y^2)}\right\vert 
= 2\left\vert\frac{g^{a}_t\left(P^a_t\left(\frac{y^1+y^2}{2}\right),P^b_t\left(\frac{y^2-y^1}{2}\right)\right)
\pm g^{b}_t\left(P^a_t\left(\frac{y^1+y^2}{2}\right),P^b_t\left(\frac{y^2-y^1}{2}\right)\right)}
{c^{\alpha^0}_t\left(P^a_t\left(\frac{y^1+y^2}{2}\right),P^b_t\left(\frac{y^2-y^1}{2}\right)\right)
+c^{\beta^0}_t\left(P^a_t\left(\frac{y^1+y^2}{2}\right),P^b_t\left(\frac{y^2-y^1}{2}\right)\right)}
\right\vert\le C_0,
$$
$$
\left\vert\frac{c^2_t(y^1,y^2)}{c^1_t(y^1,y^2)}\right\vert
=\left\vert \frac{c^{\alpha^0}_t\left(P^a_t\left(\frac{y^1+y^2}{2}\right),P^b_t\left(\frac{y^2-y^1}{2}\right)\right) 
- c^{\beta^0}_t\left(P^a_t\left(\frac{y^1+y^2}{2}\right),P^b_t\left(\frac{y^2-y^1}{2}\right)\right)}
{c^{\alpha^0}_t\left(P^a_t\left(\frac{y^1+y^2}{2}\right),P^b_t\left(\frac{y^2-y^1}{2}\right)\right) 
+ c^{\beta^0}_t\left(P^a_t\left(\frac{y^1+y^2}{2}\right),P^b_t\left(\frac{y^2-y^1}{2}\right)\right)}
\right\vert\le\lambda<1, 
$$
with $c^{\alpha}$, $P^a$, $P^b$, $g^{a}$ and $g^{b}$ defined in (\ref{eq.F.c.def}), (\ref{eq.Pa.def}), (\ref{eq.Pb.def}), (\ref{eq.ga.def}) and (\ref{eq.gb.def}). 
The first inequality holds with $C_0=5C_p$, and it follows from the boundedness of $P^a$, $P^b$ and the jump sizes. The second one follows from Assumption \ref{ass:bdd.range}.
The above inequalities imply:
$$
\left\vert \frac{\tilde{c}^2_t\Psi_{C_1^2}(Y^{2c}_t)
+ \tilde{g}_t^1}{\tilde{c}^1_t}\right\vert\le \lambda C_1^2 + C_0,
$$
for all $t$ and a.e. $\omega$. The latter estimate, together with the following lemma, imply the desired upper bound:
$$
|Y^1_t|\le \lambda C_1^2+C_0
$$
for all $t$ and a.e. $\omega$.

\begin{lemma}
Consider any constant $C>0$, any continuous function $S:[0,T]\rightarrow\RR$, absolutely bounded by $C$, any nonnegative continuous function $c$ on $[0,T]$, and any continuous function $g$ on $[0,T]$, satisfying $\left\vert g\right\vert\le C |c|$. For any $0\leq t\leq \tau \leq T$, denote:
$$
Y_{t,\tau}=\int_t^\tau \exp\left(-\int_t^s c(u) \text{d}u\right) g(s) \text{d}s 
+ \exp\left(-\int_t^\tau c(u) \text{d}u\right) S(\tau).
$$
Then
$$
|Y_{t,\tau}|\le C,\,\,\,\,\forall\, 0\leq t\leq\tau\leq T.
$$
\end{lemma}
\begin{proof}
For any $0\leq t\leq\tau\leq T$, we have
\begin{multline*}
\left\vert
\int_t^\tau \exp\left(-\int_t^s c(u) du\right) g(s) ds + \exp\left(-\int_t^\tau c(u) du\right) S(\tau)
\right\vert\\
\le
-\int_t^\tau C d \left(\exp\left(-\int_t^s c(u) du\right)\right)
+ \exp\left(-\int_t^\tau c(u) du\right)C
=C
\end{multline*}
\qed
\end{proof}

Thus, we have a solution $Y^1$ of (\ref{1c}) which satisfies $|Y^1_t|\le \lambda C_1^2+C_0$, $\PP$-a.s., for all $t$. 
Then, for $C_1^1\ge\lambda C_1^2+C_0$, we have $\Psi_{C_1^1}(Y^1_t)=Y^1_t$, and, hence, $Y^1$ also solves
\begin{equation*}\label{1cc}
\left\{
\begin{array}{l}
-\text{d}Y^1_t=\left(-c^1_t(Y^1_t,Y^{2c}_t) \Psi_{C_1^1}(Y^1_t) 
+ c^2_t \Psi_{C_1^2}(Y^{2c}_t) + g^1_t(Y^1_t,Y^{2c}_t)\right)\text{d}t 
- Z^1_t\text{d}W_t + \text{d}K_t \\
Y^1_t\ge0,\,\,\,\,\,\,\,\int_0^T Y^1_t\text{d}K_t=0,\quad Y^1_T=0 
\end{array}
\right.
\end{equation*}
Note that the above RBSDE coincides with the $Y^1$-equation in (\ref{capped}).
This one-dimensional RBSDE has a globally Lipschitz generator and, thus, a unique solution. This implies that $Y^1=Y^{1c}$, and we obtain the desired bound on $Y^{1c}$:
$$
|Y^{1c}_t|\le \lambda C_1^2+C_0,
$$
$\PP$-a.s. for all $t$, provided $C_1^1\ge\lambda C_1^2+C_0$.
Similarly, considering the $Y^2$ part of the capped system (\ref{capped}), with $Y^{1c}$ fixed, we obtain
$$
|Y^{2c}_t|\le \lambda C_2^1+C_0,
$$
$\PP$-a.s. for all $t$, provided $C_2^2\ge \lambda C_2^1+C_0$.

\noindent\emph{Step 3: Solution of the appropriately capped system solves the original system.} To show that the solution $(Y^{1c}_t,Y^{2c}_t)$ of (\ref{capped}) also solves the original system (\ref{eq.Y1.Y2.RBSDE}), we only need to show that, given the bounds on $(Y^{1c},Y^{2c})$, the capped system's generator coincides with the original generator, which translates into
\begin{equation*}
\Psi_{C_1^1}(Y^{1c}_t)=Y^{1c}_t,\,\,\,\,\,\,
\Psi_{C_2^2}(Y^{2c}_t)=Y^{2c}_t,\,\,\,\,\,\,\,
\Psi_{C_1^2}(Y^{2c}_t)=Y^{2c}_t,\,\,\,\,\,\,\,
\Psi_{C_2^1}(Y^{1c}_t)=Y^{1c}_t.
\end{equation*}
The first two equalities are satisfied if 
$$
C_1^1\ge \lambda C_1^2+C_0,\quad C_2^2\ge \lambda C_2^1+C_0,
$$
while the last two require
$$
\lambda C^1_2+C_0\le C^2_1,\quad \lambda C^2_1+C_0\le C^1_2.
$$
One can check these inequalities have a solution, as long as $\lambda<1$. The ``minimal" solution being 
$$
C_1^1=C_2^1=C^2_2=C^2_1=\frac{C_0}{1-\lambda}.
$$
With the above choice of capping, the solution to (\ref{capped}) also solves (\ref{eq.Y1.Y2.RBSDE}), thus, showing the existence of a solution of (\ref{eq.Y1.Y2.RBSDE}). 
This solution is bounded by construction. 
The uniqueness of a bounded solution follows from the fact that, when $(y^1,y^2)$ vary over a bounded set, the generator of (\ref{eq.Y1.Y2.RBSDE}) is Lipschitz, hence, the standard results yield uniqueness.
\qed\qed
\end{proof}



\begin{remark}
The above proof provides an existence result for any system (\ref{eq.Y1.Y2.RBSDE}), whose generator is given by (\ref{eq.G1.def})--(\ref{eq.G2.def}), with arbitrary (bounded and Lipschitz) progressively measurable random functions $\{c^i,g^i\}$, as long as the following holds for a.e. $(t,\omega)$ and all $(y^1,y^2)\in\RR^2$:
$$
\sum_{i=1}^2\left|g^i_t(y^1,y^2)\right| \leq C_0 c^1_t(y^1,y^2),
\quad \left|c^2_t(y^1,y^2)\right| \leq \lambda c^1_t(y^1,y^2),
$$
with some constants $C_0>0$ and $\lambda\in(0,1)$.
\end{remark}

\section{Equilibrium in the continuum-player game.}
\label{se:main}

In this section we construct an equilibrium for the continuum-player game described in Section \ref{se:setup}, in the sense of Definition \ref{def:equil.def}. The main difficulty in constructing the equilibrium stems from the mixed control-stopping nature of the game (and, of course, the fact there are multiple participants). Therefore, we attempt to break the problem into two parts - isolating the ``stopping" part of the game. In order to do this, it is convenient to make assumptions that guarantee the existence of the so-called ``extremal" agents on each side of the book. These agents are called ``extremal", because their beliefs dominate the beliefs of the other agents on the same side of the book, in the sense explained below. We denote the extremal beliefs on the long side by $\alpha^0$, and, on the short side, by $\beta^0$. In short, the agents with beliefs $\alpha^0$ are the least bullish among the long ones, and the agents with beliefs $\beta^0$ are the least bearish among the short ones.
The extremal agents, e.g., can be interpreted as market-makers, as they are closer to being market-neutral than any other agent on the same side of the book (recall that we, still, do not have any designated market makers in this game -- market-neutrality is only one of the characteristics of a market maker). Indeed, if one assumes that the long agents are bullish (which is natural, as, before the end of the game, the long agents choose to wait instead of submitting market orders), then, the agents with beliefs $\alpha^0$ are the least bullish ones.
In this section, we construct an equilibrium in which the time of the first internal market order and the bid and ask prices are determined by the extremal agents, while the rest of the shape of the LOB is due to the other agents' actions. The construction of an equilibrium, thus, splits into two parts. In the first part, the extremal agents find an equilibrium among themselves, using the results of the auxiliary two-player game, and determining the time of the first internal market order $\tau$ and the bid and ask prices $p^a$ and $p^b$. In the second part, the other agents, taking $(p^a,p^b,\tau)$ as given, determine their optimal actions. Of course, we, ultimately, prove that the strategy of every agent is optimal in the overall market, consisting of both extremal and non-extremal agents.
The resulting LOB $\nu$ has two atoms -- at the bid and ask prices -- comprised of the limit orders of the extremal and some of the non-extremal agents. The rest of the LOB contains limit orders of the non-extremal agents only. 

In order to implement the above program, we assume that $\mathbb{A}=\{\alpha^0\}\cup\hat{\mathbb{A}}$ and $\mathbb{B}=\{\beta^0\}\cup\hat{\mathbb{B}}$. 
We assume that Assumptions \ref{ass:indep.underPalpha}--\ref{ass:bdd.range} hold throughout this section. In addition, we make the following assumptions.

\begin{ass}\label{ass.intensities.comparison}
For any $\alpha\in\hat{\mathbb{A}}$, $\beta\in\hat{\mathbb{B}}$ and a.e. $(t,\omega)$, we have:
$$
\lambda^{\alpha}_tF^{+,\alpha}_t(p)\ge\lambda^{\alpha^0}_tF^{+,\alpha^0}_t(p),
\quad \lambda^{\beta}_tF^{+,\beta}_t(p)\le\lambda^{\beta^0}_tF^{+,\beta^0}_t(p),
\quad\forall p\ge0,
$$
$$
\lambda^{\alpha}_tF^{-,\alpha}_t(p)\le\lambda^{\alpha^0}_tF^{-,\alpha^0}_t(p),
\quad \lambda^{\beta}_tF^{-,\beta}_t(p)\ge\lambda^{\beta^0}_tF^{-,\beta^0}_t(p),
\quad\forall p\le0.
$$
\end{ass}

\begin{ass}\label{pa_alpha_vs_alpha_0}
For any $\alpha\in\hat{\mathbb{A}}$, $\beta\in\hat{\mathbb{B}}$ and a.e. $(t,\omega)$, we have:
$$
\frac{F^{+,\alpha_0}_t(p)}{f^{\alpha_0}_t(p)}\le \frac{F^{+,\alpha}_t(p)}{f^{\alpha}_t(p)},
\quad \frac{F^{-,\beta_0}_t(-p)}{f^{\beta_0}_t(-p)}\leq \frac{F^{-,\beta}_t(-p)}{f^{\beta}_t(-p)}
\quad\forall\,p\ge0.
$$
\end{ass}

Assumption \ref{ass.intensities.comparison} ensures that the distribution of the fundamental price at any time $t$, from an $\alpha$-agent's perspective, dominates stochastically the respective distribution from the $\alpha^0$-agent's perspective. The opposite relation holds for the short agents.
The first inequality in Assumption \ref{pa_alpha_vs_alpha_0} ensures that $\log F^{+,\alpha_0}_t(\cdot)$ decays faster than $\log F^{+,\alpha}_t(\cdot)$, which is also consistent with the interpretation that $\alpha^0$-agents assign smaller probabilities to the large jumps of the fundamental price, and larger probabilities to the small jumps, as compared to the $\alpha$-agents. Analogous interpretation holds for the second inequality in Assumption \ref{pa_alpha_vs_alpha_0}.
Assumption \ref{pa_alpha_vs_alpha_0} ensures that, in an empty LOB, the non-extremal agents would prefer to post their limit order further away from zero than the extremal ones do.

\begin{lemma}\label{lemma.opt.price.levels}
Let Assumptions \ref{ass:indep.underPalpha}--\ref{pa_alpha_vs_alpha_0} hold.
Fix any $\alpha\in\hat{\mathbb{A}}$ and $\beta\in\hat{\mathbb{B}}$. Then, for a.e. $(t,\omega)$, the following holds for all $y\in\RR$: $p\mapsto (p-y) F_t^{+,\alpha}(p)$ is non-decreasing in $p\in[y,P^{a}_t(y)]$, and $p\mapsto (y-p) F_t^{-,\beta}(p)$ is non-increasing in $p\in[P^{b}_t(y),y]$.
\end{lemma}
\begin{proof}
The statement follows easily by differentiating the target functions, recalling (\ref{eq.Pa.root.def})--(\ref{eq.Pb.root.def}), and making use of Assumption \ref{pa_alpha_vs_alpha_0}.
\qed
\end{proof}

We also need to make an assumption that limits the maximum possible demand size, as viewed by the extremal agents. Namely, the extremal agents believe that the external demand can never exceed the inventory held by these agents.

\begin{ass}\label{ass.demand.size}
For $\text{Leb}\otimes\PP$-a.e. $(t,\omega)$, we have:
$$
D_t\left(-Q^+\left(f^{\alpha^0/\beta^0}_t(x)dx\right)\right) \leq \mu^a\left(\{\alpha^0\}\right),
\quad 
-D_t\left(-Q^-\left(f^{\alpha^0/\beta^0}_t(x)dx\right)\right) \leq \mu^b\left(\{\beta^0\}\right),
$$
where $Q^+$ and $Q^-$ are defined in (\ref{eq.Qpm.def}).
\end{ass}

In order to construct an equilibrium, we need to impose certain topological conditions on the space of beliefs and on the mapping $\alpha\mapsto f^{\alpha}$.

\begin{ass}\label{ass:existence.top.AB}
The spaces $\hat{\mathbb{A}}$ and $\hat{\mathbb{B}}$ are compact metric spaces, with the Borel sigma-algebras on them (i.e. $\mu^a$ and $\mu^b$ are measures with respect to the Borel sigma-algebras). 
In addition, for a.e. $(t,\omega)$, the mapping $\alpha\mapsto f^{\alpha}_t$ is continuous as a mapping $\hat{\mathbb{A}}\to\mathbb{L}^1[0,C_p]$ and as a mapping $\hat{\mathbb{B}}\to\mathbb{L}^1[-C_p,0]$.
\end{ass}

Finally, we need to ensure that the demand size curve is ``not too flat".

\begin{ass}\label{ass:existence.demand}
There exists an increasing continuous (deterministic) function $\epsilon\colon [0,\infty)\rightarrow[0,\infty)$, s.t. $\epsilon(0)=0$ and, for a.e. $(t,\omega)$, $|D_t^{-1}(x)-D_t^{-1}(y)|\leq \epsilon(|x-y|)$, for all $x,y\in\RR$.
\end{ass}

Now, we proceed to construct a special class of equilibria in the continuum-player game. As announced earlier, the equilibrium is constructed by, first, solving the auxiliary two-player game, as described in Section \ref{se:2player}. In the two-player game, we assume that the two agents have beliefs $\alpha^0$ and $\beta^0$.
Thus, we consider the unique bounded solution $(Y^1,Y^2)$ to (\ref{eq.Y1.Y2.RBSDE}) and construct the associated $(V^a,V^b)$, which solve (\ref{eq.RBSDE.Va.Vb}), according to Lemma \ref{le:connect.RBSDEs}. Then, Lemma \ref{sab} implies that $(V^a,V^b)$ are the value functions of the two-player equilibrium $(\hat{p}^a,\hat{p}^b,\hat{\tau},\bar{p})$, where
$$
\hat{p}^a_t = P^a_t(V^a_t),\quad \hat{p}^b_t = P^b_t(V^b_t),
\quad \hat{\tau} = \inf\{t\in[0,T]\,:\, V^a_t = V^b_t\},
\quad \bar{p}_t = \frac{1}{2}V^a_t + \frac{1}{2}V^b_t.
$$
Let us introduce
\begin{equation}\label{eq.pa.pb.equil.def}
p^a_t = \hat{p}^a_t \bone_{\{t<\hat{\tau}\}} + \bar{p}_{\hat{\tau}} \bone_{\{t\geq\hat{\tau}\}},
\quad p^b_t = \hat{p}^b_t \bone_{\{t<\hat{\tau}\}} + \bar{p}_{\hat{\tau}} \bone_{\{t\geq\hat{\tau}\}}.
\end{equation}
Using these auxiliary quantities, we aim to construct an equilibrium for the continuum-player game, in which $(\nu,\theta)$ satisfy the following two conditions. First,
\begin{equation}\label{eq.nu.equil.def}
\nu^a_t=\mu^a(\{\alpha^0\})\delta_{p^a_t}+\bar{\nu}^a_t,\quad \nu^b_t=\mu^b(\{\beta^0\}) \delta_{p^b_t}+\bar{\nu}^b_t,
\end{equation}
with progressively measurable $\bar{\nu}^a$ and $\bar{\nu}^b$ taking values in the space of sigma-additive measures on $\RR$, such that, $\PP$-a.s., for all $t\in[0,T]$, $\bar{\nu}^{a}_t$ is supported on $[p^a_t,C_p]$ and $\bar{\nu}^{b}_t$ is supported on $[-C_p,p^b_t]$.\footnote{The components $\bar{\nu}^a$ and $\bar{\nu}^b$ are introduced for convenience, in order to indicate that $\nu^a_t(\{p^a_t\})\geq \mu^a(\{\alpha^0\})$ and $\nu^b(\{p^b_t\})\geq \mu^b(\{\beta^0\})$.}
Second,
\begin{equation}\label{eq.theta.equil.def}
\theta^a_t = \mu^a(\mathbb{A})\delta_{V^a_t},\quad \theta^b_t = \mu^b(\mathbb{B}) \delta_{V^b_t}.
\end{equation}
Note that, in such a market, we have
$$
\tau^a=\tau^b=\hat{\tau}
$$
The following theorem is the main result of this paper.

\begin{theorem}\label{th:main}
Let Assumptions \ref{ass:indep.underPalpha}--\ref{ass:existence.demand} hold. Consider any solution $(V^a,V^b)$ to (\ref{eq.RBSDE.Va.Vb}) (whose existence is guaranteed by Proposition \ref{prop:RBSDE.exist} and Lemma \ref{le:connect.RBSDEs}) and the associated $(p^a,p^b)$ given by (\ref{eq.pa.pb.equil.def}). Then, there exist progressively measurable measure-valued processes $(\nu,\theta)$ and random fields $p,v:\Omega\times[0,T]\times\mathbb{S}\rightarrow\mathcal{P}(\RR)\times\RR$, which form an equilibrium, in the sense of Definition \ref{def:equil.def}, and which satisfy (\ref{eq.nu.equil.def})--(\ref{eq.theta.equil.def}) along with
\begin{itemize}
\item $v_t(1,\alpha)=V^a_t$, $v_t(-1,\alpha)=V^b_t$, for all $(t,\omega,\alpha)$,
\item $p_t(1,\alpha^0)=p^a_t$, $p_t(-1,\beta^0)=p^b_t$, for all $(t,\omega)$.
\end{itemize}
\end{theorem}

\begin{remark}
Recall that there always exists a ``trivial" equilibrium, in which all agents stop at time zero. However, such an equilibrium is unrealistic and does not appear to be useful from a modeling perspective. The main contribution of the above result is the existence of a potentially non-trivial equilibrium, in which the duration of the game, $\hat{\tau}$, is determined by a solution to (\ref{eq.RBSDE.Va.Vb}), and there is no reason for it to be zero, in general.
The latter is confirmed by the numerical experiments in Section \ref{se:example}.
\end{remark}

\begin{remark}\label{re:inconsist.answer}
Notice that, as announced in Remark \ref{re:inconsist}, we have constructed an equilibrium, satisfying  
$$
v_t(1,\alpha)= v^a_t = V^a_t,
\,\,\,\,\,\,\,\,\,\,\,\,\,\,\,v_t(-1,\alpha)= v^b_t = V^b_t,
\,\,\,\,\,\,\,\,\,\,\,\,\,\,
\forall\, \alpha\in\mathbb{A}\cup\mathbb{B},\,(t,\omega)\in[0,T]\times\Omega.
$$
Therefore, in such an equilibrium, no agents execute market orders before the end of the game $\hat{\tau}$, and, hence, the empirical distribution $\mu$ remains constant and (\ref{eq.endog.mu.def}) holds.
\end{remark}

The remainder of this section is devoted to the proof of Theorem \ref{th:main}.
First, we show that, in a market $(\nu,\theta)$, satisfying (\ref{eq.nu.equil.def})--(\ref{eq.theta.equil.def}), it is never (strictly) optimal for the agents to post limit sell orders below the ask price or to post limit buy orders above the bid price. In addition, it is never (strictly) optimal for the agents to submit a market order before $\hat{\tau}$. To achieve this, we need to compare the value functions of the agents to $V^a$ and $V^b$, making use of Assumptions \ref{ass.intensities.comparison}, \ref{pa_alpha_vs_alpha_0}.

\begin{lemma}\label{le:opt.Nonext}
Let Assumptions \ref{ass:indep.underPalpha}--\ref{pa_alpha_vs_alpha_0} hold, and let $(\nu,\theta)$ satisfy (\ref{eq.nu.equil.def})--(\ref{eq.theta.equil.def}). 
Given any $\alpha\in\mathbb{A}$ and any admissible control $(p,\tau)$, for a long agent with beliefs $\alpha$, there exists an admissible control $p'$, s.t., $\PP$-a.s., $\text{supp}(p_t')\subset [p^a_t,\infty)$, for all $t\in[0,T]$, and $(p',\hat{\tau})$ does not decrease the objective value, i.e.
$$
J^{(\nu,\theta),(p,\tau)}(1,\alpha) \leq J^{(\nu,\theta),(p',\hat{\tau})}(1,\alpha).
$$
Similarly, given any $\beta\in\mathbb{B}$ and any admissible control $(p,\tau)$, for a short agent with beliefs $\beta$, there exists an admissible control $p'$, s.t., $\PP$-a.s., $\text{supp}(p_t')\subset (-\infty,p^b_t]$, for all $t\in[0,T]$, and $(p',\hat{\tau})$ does not decrease the objective value, i.e.
$$
J^{(\nu,\theta),(p,\tau)}(-1,\beta) \leq J^{(\nu,\theta),(p',\hat{\tau})}(-1,\beta).
$$
\end{lemma}
The proof of the above lemma is given in the appendix.
This lemma has a straight-forward but useful corollary.

\begin{cor}\label{cor:opt.Nonext}
Let Assumptions \ref{ass:indep.underPalpha}--\ref{pa_alpha_vs_alpha_0} hold, and let $(\nu,\theta)$ satisfy (\ref{eq.nu.equil.def})--(\ref{eq.theta.equil.def}). 
Given any $\alpha\in\mathbb{A}$, let $(p,\tau)$ be an optimal strategy for the long agents with beliefs $\alpha$, in the class of all admissible strategies satisfying: $\PP$-a.s. $\text{supp}(p_t)\subset[p^a_t,\infty)$, for all $t\in[0,T]$, and $\tau=\hat{\tau}$. Then $(p,\tau)$ is optimal in the class of all admissible strategies, in the sense of Definition \ref{def:opt}.
Similarly, given any $\beta\in\mathbb{B}$, let $(p,\tau)$ be an optimal strategy for the short agents with beliefs $\beta$, in the class of all admissible strategies satisfying: $\PP$-a.s. $\text{supp}(p_t)\subset(-\infty,p^b_t]$, for all $t\in[0,T]$, and $\tau=\hat{\tau}$. Then $(p,\tau)$ is optimal in the class of all admissible strategies, in the sense of Definition \ref{def:opt}.
\end{cor}

Thus, no matter which limit order strategy $p$ an agent is using, it is optimal for her to choose the following stopping threshold: 
$$
\hat{v}(s) = V^a \bone_{\{s>0\}} + V^b \bone_{\{s<0\}}.
$$
This implies that, given a LOB $\nu$ in the form (\ref{eq.nu.equil.def}) and the stopping strategy $\hat{v}$ as above, if an optimal limit order strategy $\hat{p}(s,\alpha)$ exists for any state $(s,\alpha)$, then $(\hat{p}(s,\alpha),\hat{v})$ form an optimal control for the agents in state $(s,\alpha)$, in the sense of Definition \ref{def:opt}. Moreover, in such a case, $\theta$, given by (\ref{eq.theta.equil.def}), satisfies the condition (\ref{eq.numinus.fixedpoint.def}). 
Next, we need to construct a LOB $\nu$, in the form (\ref{eq.nu.equil.def}), and the associated optimal limit order strategies for all agents, s.t. (\ref{eq.nuplus.fixedpoint.def}) is satisfied.
In particular, the following lemma, whose proof is postponed to the appendix, shows that, for any $\nu$ in the form (\ref{eq.nu.equil.def}), the strategies $(\delta_{p^a},V^a)$ and $(\delta_{p^b},V^b)$ are optimal for the extremal agents.

\begin{lemma}\label{le:2player.opt.forExt}
Let Assumptions \ref{ass:indep.underPalpha}--\ref{ass.demand.size} hold, and let $(\nu,\theta)$ satisfy (\ref{eq.nu.equil.def})--(\ref{eq.theta.equil.def}). Then, given $(\nu,\theta)$, the strategy $(\delta_{p^a},V^a)$ is optimal for a long agent with beliefs $\alpha^0$, and the strategy $(\delta_{p^b},V^b)$ is optimal for a short agent with beliefs $\beta^0$, in the sense of Definition \ref{def:opt}.
\end{lemma}

The remaining steps are carried out in the next subsection.

\subsection{Equilibrium strategies of the non-extremal agents}
\label{subse:fixedpoint}

In this subsection we construct the measure-valued processes $(\nu^a,\nu^b)$, in the form (\ref{eq.nu.equil.def}), and a progressively measurable random field $(\hat{p}_t(s,\alpha))$, such that the controls $(\hat{p}(1,\alpha),V^a)$ and $(\hat{p}(-1,\alpha),V^b)$ are optimal for the non-extremal agents with beliefs $\alpha$, long and short, respectively (recall that the optimal strategies for the extremal agents are constructed in Lemma \ref{le:2player.opt.forExt}), and the fixed-point constraint (\ref{eq.nuplus.fixedpoint.def}) is satisfied.
In view of Lemma \ref{le:opt.Nonext}, we can restrict the possible controls $p$ to the those satisfying: $\text{supp}(p_t)\subset[p^a_t,\infty)$, for all $t\in[0,T]$. It is also obvious that we can restrict the support of $p_t$ to be in $[-C_p,C_p]$.
As the stopping strategy is fixed, for any $\alpha\in\hat{\mathbb{A}}$, the objective of a long player reduces to $\bar{J}^{\alpha,(p)}_0$, where
\begin{equation*}
\bar{J}^{\alpha,(p)}_t =
\EE\Big[\int_t^{T} \exp\left(-\int_t^s \bar{c}^{\alpha}_u\left(p^a_u,p^b_u\right)du\right) \bar{h}^{\alpha,a}_s(p_s,p^a_s,p^b_s)ds
+ \exp\left(-\int_t^{T} \bar{c}^{\alpha}_u\left(p^a_u,p^b_u\right) du\right)
p^b_{\hat{\tau}} 
\vert \mathcal{F}_t\Big],
\end{equation*}
\begin{equation*}
\bar{c}_t^{\alpha}(p^a_t,p^b_t) = c_t^{\alpha}(p^a_t,p^b_t)\bone_{\{t\le\hat{\tau}\}},
\quad
\bar{h}^{\alpha,a}_t(p_t,p^a_t,p^b_t)= h^{\alpha,a}_t(p_t,p^a_t,p^b_t)\bone_{\{t\le\hat{\tau}\}},
\end{equation*}
and $c^{\alpha}$ and $h^{\alpha,a}$ defined in (\ref{eq.F.c.def}) and (\ref{eq.ha.def}).
Due to Assumptions \ref{ass.intensities.comparison} and \ref{ass.demand.size}, we have
\begin{equation*}
l^{c,b}_t(x) = \inf\left\{p>Q^-(\nu^b_t)\,:\,-D_t(p-x)>\nu^b_t((p,\infty))\right\}
= p^b_t\vee x,\quad\forall x\in \text{supp}(f^{\alpha}_t).
\end{equation*}
In addition, for any $z\geq p^a_t$,
$$
\{u>0\,:\,l^{c,a}_t(u)\geq z\} = \{u>0\,:\,u\geq z - D^{-1}_t\left(\nu^a_t([p^a_t,z)) \right)\},
$$
and, hence, for any $B\geq p^a_t$,
$$
\int_{p^a_t}^{B} f^{\alpha}_t(u) (l^{c,a}_t(u)-p^a_t) du
= \int_{0}^{l^{c,a}_t(B)-p^a_t} \int_{u + p^a_t - D^{-1}_t\left(\nu^a_t([p^a_t,p^a_t + u))\right)}^{C_p}  f^{\alpha}_t(y) dy du.
$$
The above observations allow us to simplify the objective:
$$
h^{\alpha,a}_t(p_t,p^a_t,p^b_t) 
= \lambda^{\alpha}_t \int_{p^a_t}^{\infty} \Big[
(z - p^b_t) F^{+,\alpha}_t\left( z - D^{-1}_t\left(\nu^a_t([p^a_t,z))\right) \right)
+ p^b_t F^{+,\alpha}_t (p^a_t)
$$
$$
+ \int_{p^a_t}^{z - D^{-1}_t\left( \nu^a_t([p^a_t,z)) \right)} 
f^\alpha_t(u) l^{c,a}_t(u) du 
\Big]p_t(dz)
+ 2 \lambda^{\alpha}_t p^b_t F^{-,\alpha}_t(p^b_t)
$$
$$
= \lambda^{\alpha}_t \int_{p^a_t}^{C_p} \Big[
(z - p^b_t) F^{+,\alpha}_t\left(z - D^{-1}_t\left(\nu^a_t([p^a_t,z))\right)\right)
$$
$$
+\int_{0}^{z-p^a_t} F^{+,\alpha}_t\left( u + p^a_t - D^{-1}_t\left(\nu^a_t([p^a_t,p^a_t + u))\right) \right) du
\Big]p_t(dz)
+ 2 \lambda^{\alpha}_t p^b_t F^{-,\alpha}_t(p^b_t)
+ \lambda^{\alpha}_t p^b_t F^{+,\alpha}_t (p^a_t).
$$
Notice that the above objective does not depend on $\nu^b$ (for a given $p^b$), hence, we can separate the equilibrium problems of the long and short agents (this is only true for the non-extremal agents, of course). For simplicity, we only consider the problem of the long agents -- the short agents can be treated similarly. 
Denote by $\kappa_t$ and $\hat{\nu}^a_t$ the push-forward measures of $p_t$ and $\nu^a_t$, under the mapping $x\mapsto x-p^a_t$. Clearly, the measurability property is preserved by this transformation, hence, we can reformulate the equilibrium problem as a search for $\kappa$ and $\hat{\nu}^a$, with the values in the space of measures with support in $[0,C_p]$. In the new variables, the objective takes a more convenient form. In particular, $h^{\alpha,a}_t(p_t,p^a_t,p^b_t) = \hat{h}^{\alpha,a}_t(\kappa_t,p^a_t,p^b_t)$, where
$$
\hat{h}^{\alpha,a}_t(\kappa_t,p^a_t,p^b_t)
= \lambda^{\alpha}_t \int_{0}^{C_p} \Big[
(z + p^a_t - p^b_t) F^{+,\alpha}_t\left(z + p^a_t - D^{-1}_t\left(\hat{\nu}^a_t([0,z))\right)\right)
$$
$$
+\int_{0}^{z} F^{+,\alpha}_t\left( u + p^a_t - D^{-1}_t\left(\hat{\nu}^a_t([0,u))\right) \right) du
\Big]\kappa_t(dz)
+ 2 \lambda^{\alpha}_t p^b_t F^{-,\alpha}_t(p^b_t)
+ \lambda^{\alpha}_t p^b_t F^{+,\alpha}_t (p^a_t).
$$
Note that $\bar{J}^{\alpha,(p)}$ solves a BSDE with the affine generator
$$
\hat{\mathcal{G}}^{\alpha}_t(y) = \bar{c}^{\alpha}_u\left(p^a_u,p^b_u\right) y
+ \hat{h}^{\alpha,a}_t(\kappa_t,p^a_t,p^b_t).
$$
In order to maximize $\bar{J}^{\alpha,(p)}$, it suffices to find a strategy $\kappa$ which maximizes the above generator. The latter is, in turn, equivalent to maximizing $\hat{h}^{\alpha,a}_t(\cdot,p^a_t,p^b_t)$. Thus, we need to find a progressively measurable random field $(\kappa_t(\alpha))$, with values in $\mathcal{P}(\RR)$ (with the weak topology on it), s.t., for $\mu^a$-a.e. $\alpha\in\hat{\mathbb{A}}$,
\begin{equation}\label{eq.kappa.fixedpt.def}
\kappa_t(\alpha)\in \text{argmax}_{\kappa' \in \psi}  \hat{h}^{\alpha,a}_t(\kappa',p^a_t,p^b_t)
\end{equation}
holds for $dt\times\PP$-a.e. $(t,\omega)$, where $\psi=\left\{  p\in\mathcal{P}(\Pi)\colon \supp(p)\subseteq \Pi\right\}$ and $\Pi=[0,C_p]$.
The standard BSDE results, then, imply that $\kappa(\alpha)$ is optimal for the agents in state $(1,\alpha)$, for $\mu^a$-a.e. $\alpha\in\hat{\mathbb{A}}$. If, in addition, we ensure that the fixed-point constraint (\ref{eq.nuplus.fixedpoint.def}) is satisfied (and a similar construction holds for the short agents), we obtain an equilibrium in the continuum-player game, in the sense of Definition \ref{def:equil.def}.
Notice that we can rewrite
$$
\hat{h}^{\alpha,a}_t(\kappa',p^a_t,p^b_t) 
= \lambda^{\alpha}_t \int_{\RR} F_t(\alpha,p,\hat{\nu}^a_t) \kappa'(dp)
+ 2 \lambda^{\alpha}_t p^b_t F^{-,\alpha}_t(p^b_t)
+ \lambda^{\alpha}_t p^b_t F^{+,\alpha}_t (p^a_t),
$$
\begin{equation}\label{eq.F.def}
F_t(\alpha,p,\hat{\nu}^a_t) =
(p + p^a_t - p^b_t) F^{+,\alpha}_t\left(p + p^a_t - D^{-1}_t\left(\hat{\nu}^a_t([0,p))\right)\right)
+\int_{0}^{p} F^{+,\alpha}_t\left( u + p^a_t - D^{-1}_t\left(\hat{\nu}^a_t([0,u))\right) \right) du.
\end{equation}
Assuming the extremal long agents post limit orders at $p^a$, the fixed-point constraint (\ref{eq.nuplus.fixedpoint.def}) (more precisely, the part of (\ref{eq.nuplus.fixedpoint.def}) that corresponds to the long agents) becomes:
\begin{equation}\label{eq.nuplus.fixedpoint.def.new}
\hat{\nu}^a_t([0,x]) = \mu^a(\{\alpha_0\}) + \int_{\hat{\mathbb{A}}} \kappa_t(\alpha;[0,x]) \mu^a(\text{d}\alpha),\quad \forall x\geq0.
\end{equation}
The above equations can be solved separately for different $(t,\omega)$, hence, to this end, we fix $(t,\omega)$ and omit the $t$ subscript whenever it causes no ambiguity. The statements that follow hold for a.e. $(t,\omega)$.
It turns out that it is more convenient to search for a measure
$$
K(d\alpha,dx) = \kappa(\alpha;dx) \mu^a(d\alpha),
$$ 
which is an element of $\mathcal{M}_{\mu^a}\left(\hat{\mathbb{A}}\times \Pi\right)$, the space of finite sigma-additive measures on $\hat{\mathbb{A}}\times \Pi$, with the first marginal $\mu^a$. Transition from $K$ to $\kappa$ is accomplished via the usual disintegration. 
Thus, for a.e. $(t,\omega)$, we need to find $(K,\nu)\in \mathcal{M}_{\mu^a}\left(\hat{\mathbb{A}}\times \Pi\right)\times\mathcal{M}_{\mu^a(\mathbb{A})}\left(\Pi\right)$ solving the following system
\begin{equation}\label{eq.fixedpt.def}
\left\{
\begin{array}{l}
{K\in \operatorname{argmax}_{K\in\mathcal{M}_{\mu^a}\left(\hat{\mathbb{A}}\times \Pi\right)}
\int F(\alpha,p,\nu) K(d\alpha,dp),\phantom{\frac{\frac{1}{2}}{2}}}\\
{\nu(dx) = \mu^a(\{\alpha_0\}) \delta_{0}(dx) + K\left(\hat{\mathbb{A}}\times dx\right),\phantom{\frac{\frac{1}{2}}{2}}}
\end{array}
\right.
\end{equation}
where $\mathcal{M}_{\mu^a(\mathbb{A})}\left( \Pi \right)$ is the space of finite sigma-additive measures on $\Pi$, with the total mass $\mu^a(\mathbb{A}) = \mu^a(\{\alpha_0\})+\mu^a(\hat{\mathbb{A}})$.
The above system can be formulated as a fixed-point problem, in an obvious way.
However, the main challenge in solving this problem stems from the fact that $F(\alpha,\cdot,\cdot)$ is not continuous: e.g. it may be discontinuous in $p$, if $\nu$ has atoms. 
Therefore, we replace $F$ by its ``mollified" version:
$$
\hat{F}(\alpha,p,\nu)=\sup_{p'\in\Pi} F(\alpha,p',\nu)-\left\vert p'-p \right\vert.
$$
The following lemma shows that we can replace $F$ by $\hat{F}$ in (\ref{eq.fixedpt.def}), and any solution to the new problem will solve the original one.

\begin{lemma}\label{le:existence.LipF}
For any $\alpha\in\hat{\mathbb{A}}$ and $\nu\in \mathcal{M}_{\mu^a(\mathbb{A})}\left(\Pi\right)$, the function $p\mapsto\hat{F}(\alpha,p,\nu)$ is $1$-Lipschitz in $p\in\Pi$, and
$$
\operatorname{argmax}_{p\in\Pi} \hat{F}(\alpha,p,\nu)=\operatorname{argmax}_{p\in\Pi} F(\alpha,p,\nu).
$$
\end{lemma}
\begin{proof}
For convenience, we drop the dependence on $(\alpha,\nu)$.
The first statement is clear from the definition. It is also clear that $\sup_{p\in\Pi} \hat{F}(p)=\sup_{p\in\Pi} F(p)$, and we denote this supremum by $S$. As $\hat{F}$ is continuous in $\Pi$, it achieves its supremum, hence, it suffices to show that $F(p_0)=S$, for every $p_0$ such that $\hat{F}(p_0)=S$ (note that the opposite implication is obvious). Assume the contrary, then $F(p)\le S-\varepsilon$, for some $\varepsilon>0$ and all $p\in \Pi\cap (p_0-\varepsilon,p_0+\varepsilon)$ by the upper semi-continuity of $F$. Then, we obtain $\hat{F}(p_0)\le S-\varepsilon$, which is a contradiction.
To see that $F$ is upper semi-continuous, notice that it is left-continuous, with only downward jumps, which follows directly from (\ref{eq.F.def}).
\qed
\end{proof}

Summarizing the above discussion, to find a solution to (\ref{eq.fixedpt.def}), it suffices to find a fixed point of the following correspondence
$$
\mathcal{M}_{\mu^a}\left(\hat{\mathbb{A}}\times \Pi\right) \ni K\mapsto \tilde{K}\left(\tilde{\nu}(K)\right),
$$ 
where
\begin{equation}\label{existence.defn.nu}
\tilde{\nu}(K;dx) = \mu(\{\alpha_0\})\delta_{p^a}(dx) + K(\hat{\mathbb{A}}\times dx) \in \mathcal{M}_{\mu^a(\mathbb{A})}\left(\Pi\right)
\end{equation}
is single-valued, and 
\begin{equation}\label{existence.defn.tildeK}
\tilde{K}(\nu)=\operatorname{argmax}\limits_{K\in\mathcal{M}_{\mu^a}\left(\hat{\mathbb{A}}\times \Pi\right)}
\int \hat{F}(\alpha,p,\nu) K(d\alpha,dp)\subset \mathcal{M}_{\mu^a}\left(\hat{\mathbb{A}}\times \Pi\right).
\end{equation}

\begin{proposition}\label{prop:existence.prop}
Let Assumptions \ref{ass:existence.top.AB}, \ref{ass:existence.demand} hold. Then, the correspondence $\mathcal{K}\colon K\mapsto \tilde{K}\left(\tilde{\nu}(K)\right)$, defined by (\ref{existence.defn.nu})--(\ref{existence.defn.tildeK}), has a fixed point.
\end{proposition}
\begin{proof}
To prove the proposition, we use the Kakutani's theorem for correspondences (cf. Definition $II.7.8.1$ and Theorem $II.7.8.6$ in \cite{fixedpoint}). 
Note that $\mathcal{M}_{\mu^a}\left(\hat{\mathbb{A}}\times \Pi\right)$, equipped with the weak topology, is convex and compact (by Prokhorov's theorem). In addition, it can be viewed as a subspace of the dual of the space of continuous functions on $\hat{\mathbb{A}}\times \Pi$, which is semi-normed. Thus, in order to apply the Kakutani's theorem, it only remains to show that $\mathcal{K}$ is upper hemi-continuous (uhc), with nonempty compact convex values.
Notice also that $\tilde{K}(\nu)$ is convex by definition (as an $\operatorname{argmax}$ of a linear functional on a convex set), hence, $\mathcal{K}$ is convex-valued, and we only need to show that it is uhc, with non-empty compact values. As $p\mapsto\tilde{\nu}(p)$ is a continuous function, and a composition of a continuous function and a uhc correspondence is a uhc correspondence, it suffices to verify that $\nu\mapsto \tilde{K}(\nu)$ is a uhc non-empty compact valued correspondence. To achieve this, we use the classical Berge's theorem (cf. $\cite{real}$, section E.3), which reduces to problem to the continuity of the functon
\begin{equation}\label{eq.phi.def}
(K,\nu)\mapsto \phi(K,\nu)=\int \hat{F}(\alpha,p,\nu) K(d\alpha,dp),
\end{equation}
on $\mathcal{M}_{\mu^a}\left(\hat{\mathbb{A}}\times \Pi\right)\times\mathcal{M}_{\mu^a(\mathbb{A})}\left(\Pi\right)$, metrized via the L\'evy-Prokhorov metric.
In the remainder of the proof, we show that $\phi(K,\nu)$ is jointly continuous in $(K,\nu)$. More precisely, $\phi(K,\nu)$ is continuous in $K$, and it is continuous in $\nu$ (with respect to L\'evy-Prokhorov metric), uniformly over $K$.

First, we show that $\phi(K,\nu)$ is continuous in $K$. By the definition of weak topology, the desired continuity would follow from the joint continuity of $\hat{F}(\alpha,p,\nu)$ with respect to $(\alpha,p)$. Due to Lemma \ref{le:existence.LipF}, $\hat{F}(\alpha,p,\nu)$ is $1$-Lipschitz in $p$ (uniformly over $\alpha\in\hat{\mathbb{A}}$), hence, it suffices to check that $\hat{F}(\alpha,p,\nu)$ is continuous in $\alpha$. The latter follows from the fact that $F(\alpha,p,\nu)$ is continuous in $\alpha$, uniformly over $p\in\Pi$. Indeed, notice that, if, for some $\alpha'\in U(\alpha)$, we have $\left\vert F(\alpha',p,\nu)-F(\alpha,p,\nu)\right\vert\le \varepsilon$ $\forall p\in \Pi$, then
$$
\hat{F}(\alpha',p,\nu)=F(\alpha',p',\nu)-\left\vert p'-p\right\vert\le F(\alpha,p',\nu)-\left\vert p'-p\right\vert+\varepsilon\le \hat{F}(\alpha,p,\nu)+\varepsilon,
$$
which, together with the analogous symmetric inequality, shows that $\left\vert \hat{F}(\alpha',p,\nu)-\hat{F}(\alpha,p,\nu) \right\vert\le\varepsilon$. The first equality in the above follows from the fact that $F$ is upper semi-continuous in $p$ (and bounded from above by $2C_p$), which is shown in the proof of Lemma \ref{le:existence.LipF}, and, hence, the supremum in the definition of $\hat{F}$ is achieved at some $p'$. To show that $F(\alpha,p,\nu)$ is continuous in $\alpha$, uniformly over $p\in\Pi$, we recall (\ref{eq.F.def}), and the desired continuity follows directly from Assumption \ref{ass:existence.top.AB}.

It remains to show that $\phi(K,\nu)$ is continuous in $\nu\in\mathcal{M}_{\mu^a(\mathbb{A})}\left(\Pi\right)$, uniformly over $K\in \mathcal{M}_{\mu^a}\left(\hat{\mathbb{A}}\times \Pi\right)$. 
As every such $K$ has a fixed finite total mass, due to the definition of $\phi$, the desired continuity follows from the fact that $\hat{F}(\alpha,p,\nu)$ is continuous in $\nu$, uniformly over $(\alpha,p)\in\hat{\mathbb{A}}\times\Pi$. To prove the latter, fix $\varepsilon>0$, and let $d_0$ be L\'evy-Prokhorov metric on $\mathcal{M}_{\mu^a(\mathbb{A})}\left(\Pi\right)$. Let us show that there exists an increasing continuous deterministic function $C_0\colon[0,\infty)\rightarrow[0,\infty)$, s.t. $C_0(0)=0$ and
$$
\left\vert \hat{F}(\alpha,p,\nu_1)-\hat{F}(\alpha,p,\nu_2)\right\vert\le C_0(\varepsilon),
\quad \forall\, p\in\Pi,\,\,\alpha\in\hat{\mathbb{A}},\,\,d_0(\nu_1,\nu_2)\le\varepsilon. 
$$
If we manage to show that there exists an increasing continuous deterministic function $B\colon[0,\infty)\rightarrow[0,\infty)$, s.t. $B(0)=0$ and
\begin{equation}\label{flpest}
F\left(\alpha,p,\nu_1\right)\le F\left(\alpha,(p-\varepsilon)\vee 0,\nu_2\right)+B(\varepsilon),
\end{equation}
then
\begin{multline*}
\hat{F}(\alpha,p,\nu_1)=F(\alpha,p',\nu_1)-\left\vert p'-p\right\vert\le
F\left(\alpha,(p'-\varepsilon)\vee0,\nu_2\right)-\left\vert p'-p\right\vert+B(\varepsilon)\\
\leq F\left(\alpha,(p'-\varepsilon)\vee 0,\nu_2\right)-\left\vert (p'-\varepsilon)\vee 0 - p\right\vert
+ B(\varepsilon) + \varepsilon\le
\hat{F}(\alpha,p,\nu_2) + B(\varepsilon) + \varepsilon.
\end{multline*} 
The latter, together with the analogous inequality in which $\nu_1$ and $\nu_2$ are switched, yields the desired uniform continuity of $\hat{F}$ in $\nu$. 
Thus, it is only left to prove (\ref{flpest}). 
For any $p\in\Pi$, by the definition of the L\'evy-Prokhorov metric, we have:
$$
\nu_1([0,p))\ge \nu_2([0,(p-\varepsilon)\vee 0))-\varepsilon
$$
and, hence, by Assumption \ref{ass:existence.demand},
$$
-D^{-1}(\nu_1([0,p)) ) \ge -D^{-1} \left( \nu_2([0,(p-\varepsilon)\vee 0))  \right)-\epsilon(\varepsilon).
$$
Then, for any $p\in\Pi$,
$$
p + p^a -D^{-1}( \nu_1([0,p)) )\ge (p-\varepsilon)\vee 0 + p^a - D^{-1} \left( \nu_2([0,(p-\varepsilon)\vee 0))  \right) - \epsilon(\varepsilon),
$$
which implies
$$
F^{+,\alpha}\left( p + p^a - D^{-1}( \nu_1^+(p))\right)\le F^{+,\alpha}\left( (p-\varepsilon)\vee 0 + p^a - D^{-1}( \nu_2^+((p-\varepsilon)\vee 0))\right) + M_f \epsilon(\varepsilon),
$$
where we used the fact that $f^{\alpha}$ is bounded by some constant $M_f$. 
The above estimate, along with the boundedness of $p^a$, $p^b$ and $F^{+,\alpha}$, yields the desired inequality (\ref{flpest}) for the first term in (\ref{eq.F.def}). Integrating the above estimate, we obtain the analogous inequality for the last term in the right hand side of (\ref{eq.F.def}), thus, completing the proof.
\qed
\end{proof}

Proposition \ref{prop:existence.prop} implies that, for a.e. $(t,\omega)$, we can find $K_{t,\omega}\in \mathcal{M}_{\mu^a}\left(\hat{\mathbb{A}}\times \Pi\right)$, s.t.
$$
K_{t,\omega} \in \tilde{K}\left(\tilde{\nu}(K_{t,\omega})\right),
$$
and, hence, $(K_{t,\omega}, \tilde{\nu}(K_{t,\omega}))$ satisfies (\ref{eq.fixedpt.def}).
Next, we need to establish the measurability of $K_{t,\omega}$ with respect to $(t,\omega)$.
Namely, we need to show that there exists a progressively measurable mapping $(t,\omega)\mapsto K_{t,\omega}\in \mathcal{M}_{\mu^a}\left(\hat{\mathbb{A}}\times \Pi\right)$, such that 
\begin{equation}\label{eq.Ktomega.fixedpt}
K_{t,\omega}\in\text{argmax}_{K'\in \mathcal{M}_{\mu^a}\left(\hat{\mathbb{A}}\times \Pi\right)} \phi_{t,\omega}\left(K',\tilde{\nu}(K_{t,\omega})\right),
\end{equation}
for $\text{Leb}\otimes\PP$-a.e. $(t,\omega)$, where $\phi$ and $\tilde{\nu}$ are defined in (\ref{eq.phi.def}) and (\ref{existence.defn.nu}).
We denote $S=[0,T]\times\Omega$, and let $\mathcal{S}$ be the progressive sigma-algebra (defined w.r.t. the filtration $\FF$) on $S$.
We also denote $\mathbb{X}=\mathcal{M}_{\mu^a}\left(\hat{\mathbb{A}}\times \Pi\right)$ and introduce the correspondence $g_1\colon S\times \mathbb{X}\to \mathbb{X}$, given by
\begin{gather*}
(t,\omega,K)\mapsto \text{argmax}_{K'\in \mathbb{X}} \phi_{t,\omega}(K',\tilde{\nu}(K)).
\end{gather*}
Notice that $\mathbb{X}$ is separable and metrizable, and consider the function $(t,\omega,K,K')\mapsto \phi_{t,\omega}(K',\tilde{\nu}(K))$, defined on $(S\times \mathbb{X}^2,\mathcal{S}\otimes\mathcal{B}(\mathbb{X}^2))$. Note that this function is continuous in $K'$ (as shown in the proof of Proposition \ref{prop:existence.prop}) and measurable in $(t,\omega,K)$ (as it is continuous in $K$ and measurable in $(t,\omega)$, as shown in the proof of Proposition \ref{prop:existence.prop}), hence, it is a Carath\'eodory function. 
Then, the Measurable Maximum theorem (cf. Theorem 18.18 in \cite{Aliprantis}) implies that $g_1$ is a $(\mathcal{S}\otimes\mathcal{B}(\mathbb{X}))$-measurable correspondence with nonempty and compact values.
Consider another correspondence $g_2\colon S\to \mathbb{X}$, given by
\begin{gather*}
(t,\omega)\mapsto \left\{K\in \mathbb{X}\colon K\in\text{argmax}_{K'}\phi_{t,\omega}(K',\tilde{\nu}(K))\right\}.
\end{gather*}
Let us show how to measurably select from $g_2$, for $\text{Leb}\otimes\PP$-a.e. $(t,\omega)$. The standard measurable selection results (cf. Corollary 18.27 and Theorem 18.26 in \cite{Aliprantis}) imply that such a selection is possible if $g_2$ has $\mathcal{S}\otimes\mathcal{B}(\mathbb{X})$-measurable graph and non-empty values. The latter follows from Proposition \ref{prop:existence.prop}, and the former is guaranteed by the following lemma.

\begin{lemma}
The correspondence $g_2$ has a $\mathcal{S}\otimes\mathcal{B}(\mathbb{X})$-measurable graph.
\end{lemma}
\begin{proof}
Denote this graph by $\Gamma_{g_2}$. Let $I_\mathbb{X}\colon \mathbb{X}\to \mathbb{X}\times \mathbb{X}$ be given by $I_\mathbb{X}(K)= (K,K)$. Then, $\Gamma_{g_2}=\left(\text{id}\times I_\mathbb{X}\right)^{-1}(\Gamma)$, where $\Gamma\subset S\times \mathbb{X}\times \mathbb{X}$ is given by
$$
\Gamma=\left\{\left(t,\omega,K,K'\vert (t,\omega)\in S,\,K\in \mathbb{X},\,K'\in\text{argmax}_{K''\in \mathbb{X}}\phi_{t,\omega}(K'',\tilde{\nu}(K))\right)\right\}
$$
$$
\cap
\left\{(t,\omega,K,K)\vert (t,\omega)\in S,\,K\in \mathbb{X}\right\}.
$$
Clearly, $\text{id}\times I_\mathbb{X}$ is a measurable map, and the set $\left\{(t,\omega,K,K)\vert (t,\omega)\in S,\,K\in \mathbb{X}\right\}$ is measurable. Therefore, we only need to check that
$$
\left\{\left(t,\omega,K,K'\vert (t,\omega)\in S,\,K\in \mathbb{X},\,K'\in\text{argmax}_{K''\in \mathbb{X}}\phi_{t,\omega}(K'',\tilde{\nu}(K))\right)\right\}
$$
is $\mathcal{S}\otimes \mathcal{B}(\mathbb{X}^2)$-measurable.
The latter set is precisely the graph of $g_1$, and it is measurable as the correspondence $g_1$ is measurable (cf. Theorem 18.6 in \cite{Aliprantis}).
\qed
\end{proof}

Thus, we conclude that there exists a progressively measurable $K$, with values in $\mathcal{M}_{\mu^a}\left(\hat{\mathbb{A}}\times\Pi\right)$, satisfying (\ref{eq.Ktomega.fixedpt}) for $\text{Leb}\otimes\PP$-a.e. $(t,\omega)$.
It only remains to construct $\kappa$ from $K$, by disintegration.
Let us introduce $A=S\times\mathbb{\hat{A}}$, equipped with the sigma-algebra $\mathcal{S}\otimes \mathcal{B}\left(\mathbb{\hat{A}}\right)$, and the measure $\QQ$ on $A\times\Pi$, defined via $\QQ(dt,d\omega,d\alpha,dp)=K_{t,\omega}(d\alpha,dp)dt \PP(d\omega)$.
Note that the marginal distribution of $\QQ$ on $A$ is $\mu^a(d\alpha) dt \PP(d\omega)$.
Then, as the natural projection from $A\times\Pi$ to $\Pi$ has a Borel range, Theorems 5.3 and 5.4 from \cite{foundations} imply that there exists a kernel $\kappa\colon A\ni(t,\omega,\alpha)\mapsto\kappa_{t,\omega}(\alpha)\in\mathcal{P}(\Pi)$, which is a regular conditional distribution of the natural projection from $A\times\Pi$ to $\Pi$, given the natural projection from $A\times\Pi$ to $\mathbb{A}$, under $\QQ$.
Namely, for every absolutely bounded measurable $f\colon A\times\Pi\to\RR$, we have
\begin{equation}\label{eq.disintegration}
\int_{A\times\Pi} f(t,\omega,\alpha,p) K_{t,\omega}(d\alpha,dp) dt \PP(d\omega) = \int_{A\times\Pi} f(t,\omega,\alpha,p) \kappa_{t,\omega}(\alpha;dp) \mu^a(d\alpha) dt \PP(d\omega).
\end{equation}
The above property yields that $\hat{\nu}^a_{t,\omega} = \tilde{\nu}(K_{t,\omega})$ and $\kappa_{t,\omega}$ satisfy the fixed-point constraint (\ref{eq.nuplus.fixedpoint.def.new}).
It only remains to show that $\kappa$ satisfies (\ref{eq.kappa.fixedpt.def}), for $\text{Leb}\otimes\PP\otimes\mu^a$-a.e. $(t,\omega,\alpha)$. Assume that this is not the case, then, there exists a measurable set $B\subset[0,T]\times\Omega$, with positive measure, s.t. for any fixed $(t,\omega)\in B$, there exists a measurable set $C\subset \hat{\mathbb{A}}$, s.t. $\mu^a(C)>0$ and, for all $\alpha\in C$,
$$
\int_{\RR} \hat{F}_{t,\omega}(\alpha,p,\tilde{\nu}(K_{t,\omega})) \kappa_{t,\omega}(\alpha;dp)
\leq \int_{\RR} F_{t,\omega}(\alpha,p,\tilde{\nu}(K_{t,\omega})) \kappa_{t,\omega}(\alpha;dp)
$$
$$
< 
\sup_{\kappa' \in \psi} \int_{\RR} F_{t,\omega}(\alpha,p,\tilde{\nu}(K_{t,\omega})) \kappa'(dp)
= \sup_{\kappa' \in \psi} \int_{\RR} \hat{F}_{t,\omega}(\alpha,p,\tilde{\nu}(K_{t,\omega})) \kappa'(dp).
$$
The above inequality becomes non-strict for all $\alpha\in \hat{\mathbb{A}}\setminus C$.
Then, for a fixed $(t,\omega)\in B$, we can choose a measurable $\tilde{\kappa}\colon \hat{\mathbb{A}} \rightarrow \mathcal{P}(\Pi)$ (in the same way as we chose a measurable $K$, except that, in this case, the measurability is required in the $\alpha$-variable), s.t.
$$
\sup_{\kappa' \in \psi} \int_{\RR} \hat{F}_{t,\omega}(\alpha,p,\tilde{\nu}(K_{t,\omega})) \kappa'(dp)
= \int_{\RR} \hat{F}_{t,\omega}(\alpha,p,\tilde{\nu}(K_{t,\omega})) \tilde{\kappa}(\alpha;dp),
\quad \mu^a\text{-a.e.}\,\alpha\in\hat{\mathbb{A}}.
$$
Thus, we obtain
$$
\int_{\RR} \hat{F}_{t,\omega}(\alpha,p,\tilde{\nu}(K_{t,\omega})) \kappa_{t,\omega}(\alpha;dp)
<  \int_{\RR} \hat{F}_{t,\omega}(\alpha,p,\tilde{\nu}(K_{t,\omega})) \tilde{\kappa}(\alpha;dp),
$$
for all $\alpha\in C$, and the non-strict inequality holds for all $\alpha\in \hat{\mathbb{A}}$. Integrating with respect to $\mu^a$, and using (\ref{eq.disintegration}) with  $f(t,\omega,\alpha,p)=\hat{F}\left(t,\omega,\alpha,p,\tilde{\nu}(K_{t,\omega}))\right)$, we obtain a contradiction with (\ref{eq.Ktomega.fixedpt}) on the set $B$ (which has a positive measure).
Thus, for $\mu^a$-a.e. $\alpha\in\hat{\mathbb{A}}$, (\ref{eq.kappa.fixedpt.def}) holds for $\text{Leb}\otimes\PP$-a.e. $(t,\omega)$. This means that, if we define $\hat{p}_t(\alpha)$ as the push-forward of $\kappa_t(\alpha)$, under the mapping $x\mapsto x+ p^a_t$, the resulting strategy $\hat{p}(\alpha)$ maximizes the generator $\hat{\mathcal{G}}^{\alpha}_t(y)$, for any $y$ and a.e. $(t,\omega)$. Then, we define $\nu^a_t$ to be the push-forward of $\hat{\nu}^a_t$, under the mapping $x\mapsto x+p^a_t$, and use the standard BSDE results to conclude that, for $\mu^a$-a.e. $\alpha\in\hat{\mathbb{A}}$,
$$
J^{(\nu,\theta),(\hat{p}(\alpha),V^a)}(1,\alpha) 
= \bar{J}^{\alpha,(\hat{p}(\alpha))}_0 
\geq \bar{J}^{\alpha,(p')}_0
= J^{(\nu,\theta),(p',V^a)}(1,\alpha)
$$
holds for all admissible strategies $p'$, which means that $\hat{p}(\alpha)$ is optimal for the long agents with beliefs $\alpha$. 
With such a choice of $\nu^a$ and $\hat{p}$, the fixed-point condition on $\nu^a$, given in (\ref{eq.nuplus.fixedpoint.def}), is satisfied, as it is equivalent to (\ref{eq.nuplus.fixedpoint.def.new}) (assuming the extremal long agents post limit orders at $p^a$, which is optimal for them). This, along with Corollary \ref{cor:opt.Nonext}, implies that $(\hat{p}(\alpha),V^a)$ is an optimal strategy for the long agents with beliefs $\alpha\in\hat{\mathbb{A}}$. The short agents are treated similarly. Thus, we complete the proof of Theorem \ref{th:main}.

\section{Example}
\label{se:example}

In this section, we consider the simplest concrete example of our model and show how it can be used. Consider a stochastic basis $(\Omega,\tilde{\FF} = \left(\mathcal{F}_t\right)_{t\in[0,T]},\PP)$, with a Poisson random measure $N$, whose compensator is $\lambda_t f_t(x)\text{d}x\text{d}t$, as described in Subsection \ref{subse:preliminaries}. We assume that $J_t(x)=x$ (i.e. $M\equiv N$), so that $N$ is the jump measure of the (potential) fundamental price process $X$.
We also assume that $T=20$, $\lambda_t\equiv1$ and $f_t$ is the density of a uniform distribution on $[-C_0,C_0]$, where the constant $C_0$ is chosen to be sufficiently large, so that this interval contains the supports of all $f^{\alpha}$ described below. 
We take $\mathbb{A}=\{\alpha_0\}\cup\hat{\mathbb{A}}$, $\mathbb{B}=\{\beta_0\}\cup\hat{\mathbb{B}}$, where 
$$
\hat{\mathbb{A}}=\left\{\frac{i}{K}\vert 0\le i <K \right\},
\quad \hat{\mathbb{B}}=\left\{-\frac{i}{K}\vert 0\le i <K \right\}
$$ 
are the uniform partitions of unit intervals, and $K=500$ is used for most of the computations herein. 
The restrictions of $\mu^a$ (resp. $\mu^b$) on $\hat{\mathbb{A}}$ (resp. $\hat{\mathbb{B}}$) assign a mass of $1/K$ to every point of the corresponding discrete space. Note that this implies $\mu^a(\hat{\mathbb{A}})=\mu^b(\hat{\mathbb{B}})=1$. We also define $\mu^a(\{\alpha^0\})=\mu^b(\{\beta^0\})=0.1$.

Next, we consider a collection of positive numbers $\{\lambda^{+,\alpha},\lambda^{-,\alpha},C^{+,\alpha},C^{-,\alpha}\}_{\alpha\in\mathbb{A}\cup\mathbb{B}}$, and define
$$
f^{\alpha}(x)=\frac{\lambda^{+,\alpha}}{(\lambda^{+,\alpha}+\lambda^{-,\alpha})C^{+,\alpha}} \bone_{[0,C^{+,\alpha}]}(x) 
+ \frac{\lambda^{-,\alpha}}{(\lambda^{+,\alpha}+\lambda^{-,\alpha})C^{-,\alpha}} \bone_{[-C^{-,\alpha},0]}(x),
\quad \lambda^{\alpha}=\lambda^{+,\alpha}+\lambda^{-,\alpha}.
$$
Herein, we use  $C^{+,\alpha_0}=C^{-,\alpha_0}=C^{+,\beta_0}=C^{-,\beta_0}=0.5$ and
$$
C^{+,\alpha}=a+b\alpha,
\quad C^{-,\alpha}=C^{-,\alpha_0},
\quad \forall\,\alpha\in\hat{\mathbb{A}},
\quad C^{-,\beta}=a-b\beta,
\quad C^{+,\beta}=C^{+,\beta_0},
\quad \forall\,\beta\in\hat{\mathbb{B}},
$$ 
with $a=0.5$ and $b=10$.
Finally, for any $\alpha\in\mathbb{A}\cup\mathbb{B}$, we introduce
$$
\Gamma^{\alpha}(x)=\frac{\lambda^{\alpha}}{\lambda}\frac{f^\alpha(x)}{f(x)}-1,
\quad \text{d}Z^\alpha_t=Z^\alpha_{t-}\int_{\mathbb{R}} \Gamma^{\alpha}(x)\,[{N}(\text{d}t,\text{d}x)-\lambda f(x)dtdx],
$$
and define $\PP^{\alpha}<<\PP$ by its Radon-Nikodym density $Z^{\alpha}_T$. One can easily check, using the general results in \cite{JacodShiryaev} (or in \cite{ContTankov}, for the deterministic case, used herein) that, under such $\PP^{\alpha}$, $N$ is a Poisson random measure with the compensator $\lambda^{\alpha}f^{\alpha}(x)\text{d}x\text{d}t$. 

We assume that the demand elasticity is deterministic, constant in time, and linear in price:
$$
D_t(p)=-kp,
$$
with the elasticity parameter $k=0.2$. 
With the above choice of $(C^{\pm,\alpha_0},C^{\pm,\beta_0},\mu^a(\{\alpha^0\}),\mu^b(\{\beta^0\}),k)$, it is easy to see that Assumption \ref{ass.demand.size} is satisfied.
Notice that the choice of $\lambda^{\pm,\alpha}$, for $\alpha\in\hat{\mathbb{A}}\cup\hat{\mathbb{B}}$, does not affect the equilibrium, as long as Assumptions \ref{ass.intensities.comparison} and \ref{pa_alpha_vs_alpha_0} are satisfied. This is, clearly, the case if we choose $\lambda^{\pm,\alpha}=\lambda^{\pm,\alpha^0}$ and $\lambda^{\pm,\beta}=\lambda^{\pm,\beta^0}$, for $\alpha\in\mathbb{A}$ and $\beta\in\mathbb{B}$. Herein, we consider several different sets of values for $(\lambda^{\pm,\alpha^0},\lambda^{\pm,\beta^0})$.

Let us construct an equilibrium in this example. Notice that, in the present case, the Brownian motion $W$ does not affect the jump intensities and, in turn, the agents' objectives, hence, the RBSDE system (\ref{eq.Y1.Y2.RBSDE}) becomes a system of reflected ODEs. We can solve it easily, using a simple Euler scheme, then, recover the value functions $(V^a,V^b)$, as shown in Lemma \ref{le:connect.RBSDEs}, and construct the bid and ask prices, $(p^a,p^b)$, in the feedback form, as shown in Lemma \ref{sab}.
We implement this strategy with the parameters chosen above, and with $\lambda^{+,\alpha_0}=2.5$, $\lambda^{-,\alpha_0}=1$, $\lambda^{+,\beta_0}=1$, $\lambda^{-,\beta_0}=2.5$ (so that the extremal ask agents are bullish whereas the extremal bid agents are bearish). The results are shown in the left part of Figure \ref{fig:1}.
Using the same parameters, we consider the book beyond the best bid and ask prices. In order to construct it, we solve the fixed-point problem (\ref{eq.fixedpt.def}) numerically. The latter is achieved by limiting the set of possible price levels for the limit orders to a finite set (i.e. to a partition of a large interval), which reduces (\ref{eq.fixedpt.def}) to a finite-dimensional fixed-point problem. In addition, we allow each agent to post a limit order at a single price level only, which further simplifies the problem.\footnote{Note that this restriction does not compromise the optimality of the agents' actions, provided a fixed point can be found. Indeed, it is a well known phenomenon that, in a continuum-player game, an equilibrium with pure controls also provides an equilibrium for a setting with distributed controls. This is, in fact, one of the advantages of the continuum-player games. We consider distributed controls only to prove that the equilibrium does exist, which is much harder (if at all possible) to show for a setting with pure controls.} Thus, we find a solution by the standard recursive iteration, maximizing, at each step, the objective over a finite set. The resulting optimal limit order strategies of the agents (at time zero) are plotted in the right part of Figure \ref{fig:1}, as a function of the agents' beliefs $\alpha\in\hat{\mathbb{A}}\cup\hat{\mathbb{B}}$. Notice that the optimal limit order strategy $p(\cdot)$ is piece-wise constant. It is worth mentioning that this discreteness seems to be inherent in the model and not just an artifact of the discretization of prices or beliefs that we chose herein, as the results do not change when we increase the number of possible beliefs ($K$) and price levels.

Finally, we demonstrate how the proposed framework can be used to model the \emph{indirect market impact}, which appears when an initial change to the LOB creates ``feedback loop" and causes further changes. Note that the initial change may be triggered by a trade (which is the case in the classical models of optimal execution) or by a new limit order. An extreme example of the latter is the so-called ``spoofing" -- i.e. posting a large limit order with the goal to make the price of the asset move in the opposite direction.\footnote{We stress that intentional spoofing is an illegal activity.} To the best of our knowledge, to date, there exists no model capable of explaining how exactly this activity causes the LOB (and, in particular, the price) to change. To model this process, we modify the present example by assuming that $(\lambda^{\pm,\alpha^0},\lambda^{\pm,\beta^0})$ are, in fact, functions of a relevant market indicator, which we denote by $I$: 
\begin{equation}\label{eq.I.to.lambda}
\lambda^{+,\alpha^0}=2.3\exp\left(I s\right),
\quad \lambda^{-,\alpha^0}=\exp\left(-I s\right),
\quad \lambda^{+,\beta^0}=\exp\left(I s\right),
\quad \lambda^{-,\beta^0}=2.3\exp\left(-I s\right),
\end{equation}
where $s=2.6$ is the sensitivity. We further assume that $I$ is the so-called \emph{market imbalance}: the ratio of the size of all limit orders at the best bid over the size of all limit orders at the best ask, less one. It is a well known empirical fact (cf. \cite{MMS27}, \cite{MMS30}, \cite{MMS31}) that such an indicator has a predictive power for the direction of the next price move. Note that $I$ is a function of the LOB, which, in turn, is an outcome of an equilibrium, in which $I$ is the input. Strictly speaking, our results do not guarantee the existence of an equilibrium with this additional fixed-point constraint.
In fact, an equilibrium with ``feedback beliefs", given by (\ref{eq.I.to.lambda}), can be viewed as a fixed-point of the following mapping:
\begin{equation}\label{eq.ex.FixedPtMapping.def}
(\lambda^{\pm,\alpha^0},\lambda^{\pm,\beta^0})\mapsto \nu \mapsto I \mapsto (\lambda^{\pm,\alpha^0},\lambda^{\pm,\beta^0}),
\end{equation}
where $(\lambda^{\pm,\alpha^0},\lambda^{\pm,\beta^0})\mapsto \nu$ maps the numbers $(\lambda^{\pm,\alpha^0},\lambda^{\pm,\beta^0})$ into an equilibrium LOB $\nu$, as it is done in the first part of this section, and $I \mapsto (\lambda^{\pm,\alpha^0},\lambda^{\pm,\beta^0})$ is given by (\ref{eq.I.to.lambda}). Herein, we do not prove a general existence result for the aforementioned fixed point, but we can compute it numerically by applying the associated mapping iteratively (assuming the iterations do converge).
In particular, the top right part of Figure \ref{fig:2} shows an example of LOB arising in equilibrium with feedback beliefs, given by (\ref{eq.I.to.lambda}), with $I_0=.0984456$.

Our next goal is to show how the market may move from one equilibrium to another, once the LOB is perturbed (assuming the feedback beliefs (\ref{eq.I.to.lambda})). It is worth mentioning that there is no canonical way to describe how agents achieve an equilibrium. Nevertheless, we propose a specific algorithm, based on the iterations of (\ref{eq.ex.FixedPtMapping.def}), with the following rationale behind it.
For any parameters $\bar{\lambda}=(\lambda^{\pm,\alpha^0},\lambda^{\pm,\beta^0})$ (given as functions of time), the agents know their equilibrium strategies: $(p(\bar{\lambda}),v(\bar{\lambda}))$, which can be computed as shown in the first part of this example (and whose existence follows from the main result of this paper). If the LOB is perturbed, $I$ changes, and, in turn, the parameters change from $\bar{\lambda}$ to $\bar{\lambda}'$, via (\ref{eq.I.to.lambda}). Then, the agents change their strategies to $(p(\bar{\lambda}'),v(\bar{\lambda}'))$, which form an equilibrium with respect to the new set of parameters $\bar{\lambda}'$. In the new equilibrium, the LOB, and, hence, the imbalance $I$, may change, causing further change to the parameters, and so on, until the agents reach a set of parameters that coincides with the previous one (or, almost coincides, from a numerical point of view). We believe that this algorithm for moving to a new equilibrium makes economic sense, although, of course, it is not the only possible choice. Mathematically, it corresponds to iterating the mapping (\ref{eq.ex.FixedPtMapping.def}).
To illustrate this approach, we add an extra limit buy order of size $0.05$, located at the best bid price, to the previously obtained equilibrium LOB -- as shown in the bottom right part of Figure \ref{fig:2}. This implies a change to the imbalance $I$ and, in turn, to the agents' parameters $(\lambda^{\pm,\alpha^0},\lambda^{\pm,\beta^0})$, via (\ref{eq.I.to.lambda}). Hence, the agents adjust their controls to reach a new equilibrium, then, re-calculate the parameters with the new imbalance, and so on.
Figure \ref{fig:3} shows what happens to the LOB and to the functions $(V^a,V^b)$ in the first five iterations. We can see that the initial change in imbalance makes the agents more bullish about the asset, and they tend to move their limit orders higher. In particular, the size of the best bid queue increases, while the size of the best ask queue decreases, further increasing the market imbalance. The left part of Figure \ref{fig:3} also shows that, starting from step three, the value functions $V^a$ and $V^b$ coincide at time zero, which means that the agents, in fact, choose to submit an internal market order, terminating the game. The latter constitutes an equilibrium with feedback beliefs (\ref{eq.I.to.lambda}).\footnote{It is important to notice that an equilibrium with feedback beliefs (\ref{eq.I.to.lambda}), typically, is not unique, but the proposed algorithm leads to a specific one. The resulting equilibrium is degenerate, in the sense that the game ends immediately, but, of course, there exist other equilibria.} This experiment, in particular, shows why the predictive power of market imbalance is a ``self-fulfilling prophecy": the fact that the agents base their beliefs about the type of the next market order on the market imbalance, itself, implies that a sufficient change in market imbalance will, indeed, trigger a market order of the right type. 

Of course, the analysis provided in the second part of this section is merely an example, which is meant to illustrate a potential application of our theoretical results. Namely, our main results show that a single iteration of the mapping (\ref{eq.ex.FixedPtMapping.def}) is well defined. Nevertheless, a rigorous analysis of the resulting iterative scheme, including its convergence to a fixed point, is missing.
In general, it would also be interesting to consider other indicators: e.g. choosing the size and direction of the last trade as the relevant indicator, would allow one to model the indirect impact of a market order on the LOB (in addition to the obvious, direct, impact, resulting from the immediate execution of limit orders). In our future research, we plan to find appropriate model specifications which would allow us to conduct a more thorough analysis of the indirect market impact, within the proposed setting, and to test the predictions of our model against the market data.

\section{Appendix}

\emph{Proof of Lemma \ref{le:opt.Nonext}}.
We consider a long agent with beliefs $\alpha$ and introduce
\begin{equation*}
\bar{J}^{\alpha,(p,\tau)}_t =
\EE\Big[\int_t^{\tau} \exp\left(-\int_t^s \bar{c}^{\alpha}_u\left(p^a_u\wedge Q^-(p_u),p^b_u\right)du\right) \bar{h}^{\alpha,a}_s(p_s,p^a_s,p^b_s)ds
\end{equation*}
\begin{equation*}
+ \exp\left(-\int_t^{\tau} \bar{c}^{\alpha}_u\left(p^a_u\wedge Q^-(p_u),p^b_u\right) du\right)
p^b_{\tau\wedge\hat{\tau}} 
\vert \mathcal{F}_t\Big],
\end{equation*}
where
\begin{equation*}
\bar{c}_t^{\alpha}(x,y) = c_t^{\alpha}(x,y)\bone_{\{t\le\hat{\tau}\}},
\quad
\bar{h}^{\alpha,a}_t(\kappa,x,y)= h^{\alpha,a}_t(\kappa,x,y)\bone_{\{t\le\hat{\tau}\}},
\quad x,y\in\RR\,\,\kappa\in\mathcal{P}(\RR),
\end{equation*}
with $c^{\alpha}$ and $h^{\alpha,a}$ defined in (\ref{eq.F.c.def}) and (\ref{eq.ha.def}).
Next, for any $t\in[0,T]$, any $\alpha\in\mathbb{A}$, and any admissible $p$, we introduce
\begin{equation}\label{eq.Y.extagents.def}
Y^{\alpha,p}_t=\operatorname{ess\,sup}\limits_{\tau\in\mathcal{T}_{t}} \bar{J}_t^{\alpha,(p,\tau)},
\end{equation}
The standard results on RBSDEs imply that $Y^{\alpha,p}$ is the unique $\mathbb{S}^2$ solution of the affine RBSDE,
\begin{eqnarray}
-dY^{\alpha,p}_t = \bar{\mathcal{G}}^{\alpha,p}_t(Y^{\alpha,p}_t)dt - Z_t\text{d}W_t + \text{d}K_t\quad0\le t\le T
\label{eq.BSDE.2player.affine.first.n}\\
Y^{\alpha,p}_t\ge p^b_{t\wedge\hat{\tau}} \quad 0\le t \le T,\quad \int_0^T (Y^{\alpha,p}_t-p^b_{t\wedge\hat{\tau}})\text{d}K_t=0 \\
Y^{\alpha,p}_T=p^b_{\hat{\tau}},
\label{eq.BSDE.2player.affine.last.n}
\end{eqnarray}
where 
$$
\bar{\mathcal{G}}^{\alpha,p}_t(y) 
= -\bar{c}^{\alpha}_t\left(p^a_t\wedge Q^-(p_t),p^b_t\right) y 
+ \bar{h}^{\alpha,a}_t(p_t,p^a_t,p^b_t)
= \left[-c^{\alpha}_t\left(p^a_t\wedge Q^-(p_t),p^b_t\right) y 
+ h^{\alpha,a}_t(p_t,p^a_t,p^b_t)\right] \bone_{\{t<\hat{\tau}\}},
$$
with $c^{\alpha}$ and $h^{\alpha,a}$ defined in (\ref{eq.F.c.def}) and (\ref{eq.ha.def}).
Recall that $V^a$ satisfies (\ref{eq.RBSDE.Va.Vb}), with the generator 
$$
\mathcal{G}^a_t(y,p^b_t) = 
2\lambda_t^{\alpha^0} p^b_t F^{\alpha^0,-}_t(p^b_t)
- \lambda^{\alpha^0}_t F^{-,\alpha^0}_t(p^b_t) y
+ \lambda_t^{\alpha^0} P^a_t(y) F^{\alpha^0,+}_t(P^a_t(y))
- \lambda^{\alpha^0}_t F^{+,\alpha^0}_t(P^a_t(y)) y.
$$
It is easy to deduce that
$$
\bar{\mathcal{G}}^{\alpha^0,p^a}_t\left(V^a_{t}\right) = \mathcal{G}^a_t(V^a_t,p^b_t) \bone_{\{t<\hat{\tau}\}}.
$$
Hence, $(V^a_{t\wedge\hat{\tau}})$ satisfies the same RBSDE as $(Y^{\alpha^0,p^a}_t)$. From the comparison principle, we conclude that $Y^{\alpha^0,p^a}_t = V^a_{t\wedge\hat{\tau}}$.
On the other hand, for any $\alpha\in\mathbb{A}$, let us choose $p_t = \delta_{p^a_t}$, to obtain:
$$
\bar{\mathcal{G}}^{\alpha,p^a}_t\left(Y^{\alpha^0,p^a}_t\right)
= \bar{\mathcal{G}}^{\alpha,p^a}_t\left(V^a_t\right) 
= \left[
\lambda^{\alpha}_t p^b_t F^{-,\alpha}_t(p^b_t)
+ \lambda^{\alpha}_t \int_{-\infty}^{p^b_t} f^\alpha_t(u) l^{c,b}_t(u) du
- \lambda^{\alpha}_t F^{-,\alpha}_t(p^b_t) V^a_t 
\right.
$$
$$
\left.
+ \lambda^{\alpha}_t p^a_t F^{+,\alpha}_t(p^a_t)
- \lambda^{\alpha}_t F^{+,\alpha}_t(p^a_t) V^a_t
\right] \bone_{\{t<\hat{\tau}\}}
$$
$$
\geq \left[
\lambda^{\alpha}_t F^{-,\alpha}_t(p^b_t) (p^b_t - V^a_t)
+ \lambda^{\alpha}_t F^{+,\alpha}_t(p^a_t) (p^a_t - V^a_t)
+ \lambda^{\alpha^0}_t p^b_t F^{-,\alpha^0}_t(p^b_t)
\right] \bone_{\{t<\hat{\tau}\}},
$$
where $l^{c,b}$ is defined in (\ref{eq.lcb.def}), and the last inequality is based on the Assumptions \ref{ass.intensities.comparison}, \ref{ass.demand.size}, and on the monotonicity of $l^{c,b}_t(\cdot)$, which imply
$$
\lambda^{\alpha}_t \int_{-\infty}^{p^b_t} f^\alpha_t(u) l^{c,b}_t(u) du
\geq \lambda^{\alpha^0}_t \int_{-\infty}^{p^b_t} f^{\alpha^0}_t(u) l^{c,b}_t(u) du
= \lambda^{\alpha^0}_t p^b_t F^{-,\alpha^0}_t(p^b_t).
$$
Notice that, by construction, $p^b_t \leq V^b_t \leq V^a_t \leq p^a_t$.
Then,  Assumption \ref{ass.intensities.comparison} implies
$$
\lambda^{\alpha}_t F^{-,\alpha}_t(p^b_t) (2p^b_t - V^a_t)
\geq \lambda^{\alpha^0}_t F^{-,\alpha^0}_t(p^b_t) (2p^b_t - V^a_t),
\quad
\lambda^{\alpha}_t F^{+,\alpha}_t(p^a_t) (p^a_t - V^a_t)
\geq \lambda^{\alpha^0}_t F^{+,\alpha^0}_t(p^a_t) (p^a_t - V^a_t).
$$
Thus, we obtain:
$$
\bar{\mathcal{G}}^{\alpha^0,p^a}_t\left(Y^{\alpha^0,p^a}_t\right) 
\leq \bar{\mathcal{G}}^{\alpha,p^a}_t\left(Y^{\alpha^0,p^a}_t\right).
$$
Using the comparison principle for RBSDEs, we conclude that $Y^{\alpha,p^a}_t\geq Y^{\alpha^0,p^a}_t = V^a_{t\wedge\hat{\tau}}$.
Consider an arbitrary strategy $(p,\tau)$. By switching between $p^a$ and $p$, we can construct a new strategy $p'$, such that $Y^{\alpha,p'}_t\geq V^a_{t\wedge\hat{\tau}}\vee Y^{\alpha,p}_t$, for all $t$. More precisely, we define
$$
\bar{\mathcal{G}}^{\alpha,p'}_t(y) = \bar{\mathcal{G}}^{\alpha,p^a}_t(y) \vee \bar{\mathcal{G}}^{\alpha,p}_t(y),
$$
and solve the RBSDE (\ref{eq.BSDE.2player.affine.first.n})--(\ref{eq.BSDE.2player.affine.last.n}). By the standard argument, the $Y$-component of the solution is $Y^{\alpha,p'}$, where $p'_t$ is defined to be equal to $\delta_{p^a_t}$ if the maximum in the above equation is achieved at $\bar{\mathcal{G}}^{\alpha,p^a}_t(Y^{\alpha,p'}_t)$, and it is equal to $p_t$ otherwise. 
The comparison principle implies that $Y^{\alpha,p'}_t \geq Y^{\alpha,p^a}_t\vee Y^{\alpha,p}_t\geq V^a_{t\wedge\hat{\tau}}\vee Y^{\alpha,p}_t$.
Then, the standard results on RBSDEs imply that the optimal stopping time associated with $Y^{\alpha,p'}$ is
$$
\inf\{t\in[0,T]\,:\,  Y^{\alpha,p'}_t \leq p^b_{t\wedge\hat{\tau}}\} 
= \inf\{t\in[0,T]\,:\,  V^{a}_t \leq p^b_{t\wedge\hat{\tau}}\} = \hat{\tau}.
$$ 
Thus,
$$
J^{(\nu,\theta),(p,\tau)}_0(1,\alpha) = \bar{J}^{\alpha,(p,\tau)}_0
\leq Y^{\alpha,p}_0 \leq Y^{\alpha,p'}_0 = \bar{J}^{\alpha,(p',\hat{\tau})}_0
= J^{(\nu,\theta),(p',\hat{\tau})}_0(1,\alpha).
$$
Next, we show that the control $p$ can be chosen so that, $\PP$-a.s., for all $t$, $\text{supp}(p_t)\subset [p^a_t,\infty)$.
Consider any control $p$. By switching, if necessary, between $p^a$ and $p$, we can ensure that $Y^{\alpha,p}_t\geq V^a_{t\wedge\hat{\tau}}$. Then, for $t<\hat{\tau}$, the generator of $Y^{\alpha,p}$ is given by
$$
\bar{\mathcal{G}}^{\alpha,p}_t(y)
= -c^{\alpha}_t\left(p^a_t\wedge Q^-(p_t),p^b_t\right) y 
+ h^{\alpha,a}_t(p_t,p^a_t,p^b_t)
$$
$$
= -\lambda^{\alpha}_t F^{+,\alpha}_t(p^a_t\wedge Q^-(p_t)) y
+ \lambda^{\alpha}_t p^b_t F^{+,\alpha}_t(Q^-(p_t)\wedge p^a_t) 
p_t\left((p^a_t,\infty)\right)
$$
$$
+ \lambda^{\alpha}_t \int_{(Q^-(p_t)\wedge p^a_t)\vee0}^{p^a_t} f^{\alpha}_t(u) 
\int_{-\infty}^{p^a_t} \left[
\left(z\wedge u + \left(p^b_t - u\right) \bone_{\{z> u\}}\right) \right] p_t(dz) du
+ \lambda^{\alpha}_t F^{+,\alpha}_t(p^a_t) 
\int_{-\infty}^{p^a_t} z p_t(dz)
$$
$$
- \lambda^{\alpha}_t F^{-,\alpha}_t(p^b_t) y 
+ 2\lambda^{\alpha}_t p^b_t F^{-,\alpha}_t(p^b_t)
-  \lambda^{\alpha}_t p^b_t F^{+,\alpha}_t(p^a_t) 
p_t\left((p^a_t,\infty)\right)
$$
$$
+ \lambda^{\alpha}_t \int_{p^a_t}^{\infty} f^{\alpha}_t(u) 
\int_{p^a_t}^{\infty} \left[ \left(z\wedge l^{c,a}_t(u) + p^b_t \bone_{\{z> l^{c,a}_t(u)\}} \right) \right] p_t(dz) du
$$
Let us estimate the first four terms in the right hand side of the above (i.e. the ones that depend on $p_t(dx)$ restricted to $x< p^a_t$):
$$
-\lambda^{\alpha}_t F^{+,\alpha}_t(p^a_t\wedge Q^-(p_t)) y
+ \lambda^{\alpha}_t p^b_t F^{+,\alpha}_t(Q^-(p_t)\wedge p^a_t) 
p_t\left((p^a_t,\infty)\right)
$$
$$
+ \lambda^{\alpha}_t \int_{(Q^-(p_t)\wedge p^a_t)\vee0}^{p^a_t} f^{\alpha}_t(u) 
\int_{-\infty}^{p^a_t} \left[
\left(z\wedge u + \left(p^b_t - u\right) \bone_{\{z> u\}}\right) \right] p_t(dz) du
+ \lambda^{\alpha}_t F^{+,\alpha}_t(p^a_t) 
\int_{-\infty}^{p^a_t} z p_t(dz)
$$
$$
\leq \lambda^{\alpha}_t \sup_{x\leq p^a_t}
\Big[
\left(- y + p^b_t \right) F^{+,\alpha}_t(x)
+ p_t\left((-\infty,p^a_t]\right) 
\sup_{z\in[x,p^a_t]} \Big[(z-p^b_t) F^{+,\alpha}_t(z) 
\Big]\Big].
$$
Notice that, for $t<\hat{\tau}$ and $y=Y^{\alpha,p}_t$, we have $p^b_t - y \leq 0$, and, hence,
$$
\sup_{x\leq p^a_t}
\Big[
\left(- y + p^b_t \right) F^{+,\alpha}_t(x)
+ p_t\left((-\infty,p^a_t]\right) 
\sup_{z\in[x,p^a_t]} \Big[(z-p^b_t) F^{+,\alpha}_t(z) 
\Big]\Big]
$$
$$
= \sup_{z\leq p^a_t}
\Big[(z - y) F^{+,\alpha}_t(z) 
+ p^b_t p_t\left((p^a_t,\infty)\right) F^{+,\alpha}_t(z) 
\Big].
$$
Due to Lemma \ref{lemma.opt.price.levels}, the function $z\mapsto (z - y) F^{+,\alpha}_t(z)$ is nondecreasing in $z\leq P^a_t(y)$. As $p^b_t\leq0$, the function $z\mapsto p^b_t\, p_t\left((p^a_t,\infty)\right) F^{+,\alpha}_t(z)$ is also nondecreasing, and, hence, the above supremum is attained at $z=p^a_t$, provided $P^a_t(y)\geq p^a_t$. The latter does hold for $t<\hat{\tau}$ and $y=Y^{\alpha,p}_t$, as $P^a_t(\cdot)$ is non-decreasing, $p^a_t=P^a_t(V^a_t)$ and $Y^{\alpha,p}_t\geq V^a_t$.
Thus, the generator $\bar{\mathcal{G}}^{\alpha,p}_t(Y^{\alpha,p}_t)$ does not decrease if we replace $p$ by 
$$
p_t'(dx) = p_t(dx)\bone_{[p^a_t,\infty)} + p_t((-\infty,p^a_t))\delta_{p^a_t}(dx).
$$
In other words, 
$$
\bar{\mathcal{G}}^{\alpha,p}_t(Y^{\alpha,p}_t) \leq \bar{\mathcal{G}}^{\alpha,p'}_t(Y^{\alpha,p}_t).
$$
The comparison principle, then, yields $Y^{\alpha,p}_t \leq Y^{\alpha,p'}_t$. Moreover, the optimal stopping strategy associated with $Y^{\alpha,p'}$ is $\hat{\tau}$.
Repeating the argument used earlier in this proof, we conclude that any strategy $(p,\hat{\tau})$ can be modified to $(p',\hat{\tau})$, satisfying the properties stated in the lemma, so that the objective value does not decrease.
The case of short agents is treated similarly.
\qed

\emph{Proof of Lemma \ref{le:2player.opt.forExt}}
Consider a long agent with beliefs $\alpha^0$. In view of Corollary \ref{cor:opt.Nonext}, it suffices to show the optimality in the class of strategies $(p,\hat{\tau})$, with $\text{supp}(p_t)\subset[p^a_t,\infty)$.
Notice that Assumption \ref{ass.demand.size} implies:
$$
l^{c,a}_t(x) = x\wedge p^a_t,\quad\forall x\in \text{supp}\left(f^{\alpha^0}_t\right).
$$
Using the above observation, we recall the constructions from the proof of Lemma \ref{le:opt.Nonext}, to obtain, for any strategy $p$ and all $t<\hat{\tau}$:
$$
\bar{\mathcal{G}}^{\alpha^0,p}_t(y)
= -\lambda^{\alpha^0}_t F^{+,\alpha^0}_t(p^a_t) y
- \lambda^{\alpha^0}_t F^{-,\alpha^0}_t(p^b_t) y 
+ 2\lambda^{\alpha^0}_t p^b_t F^{-,\alpha^0}_t(p^b_t)
$$
$$
+ \lambda^{\alpha^0}_t F^{+,\alpha^0}_t(p^a_t) \left(p^a_t p_t(\{p^a_t\})
+ (p^a_t + p^b_t) p_t\left((p^a_t,\infty)\right) \right).
$$
As $p^b_t\leq0$, the above expression is maximized at $p_t = \delta_{p^a_t}$.
Using the comparison principle for the RBSDE satisfied by $Y^{\alpha^0,p}$, we conclude that $p=\delta_{p^a}$ produces the largest $Y^{\alpha^0,p}$ and, hence, the largest objective value for the long agents with beliefs $\alpha^0$. The case of short agents is treated similarly.
\qed

\begin{figure}
\begin{center}
  \begin{tabular} {cc}
    {
    \includegraphics[width = 0.49\textwidth]{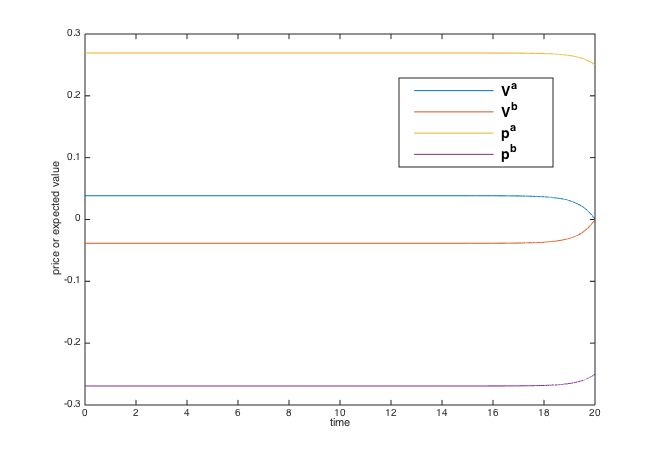}
    } & {
    \includegraphics[width = 0.49\textwidth]{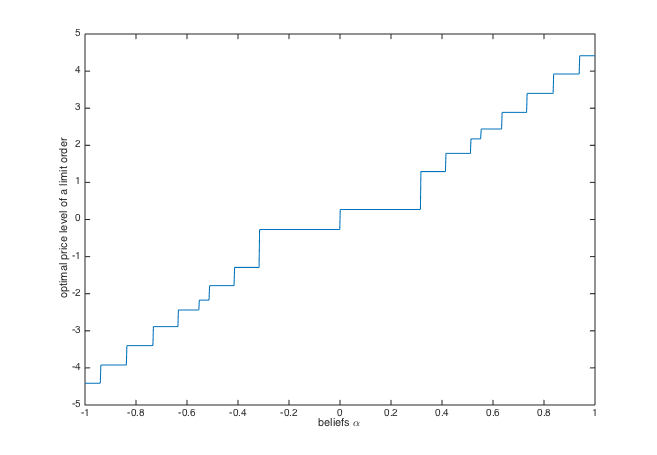}
    }\\
  \end{tabular}
  \caption{On the left: value functions $(V^b,V^a)$ (red and blue), and the bid and ask prices $(p^b,p^a)$ (purple and orange), as functions of time. On the right: the optimal price level of a limit order, as a function of the beliefs $\alpha\in\hat{\mathbb{A}}\cup\hat{\mathbb{B}}$. Parameters: $\lambda^{+,\alpha_0}=2.5$, $\lambda^{-,\alpha_0}=1$, $\lambda^{+,\beta_0}=1$, $\lambda^{-,\beta_0}=2.5$.}
    \label{fig:1}
\end{center}
\end{figure}

\begin{figure}
\begin{center}
  \begin{tabular} {cc}
   {
    \includegraphics[width = 0.49\textwidth]{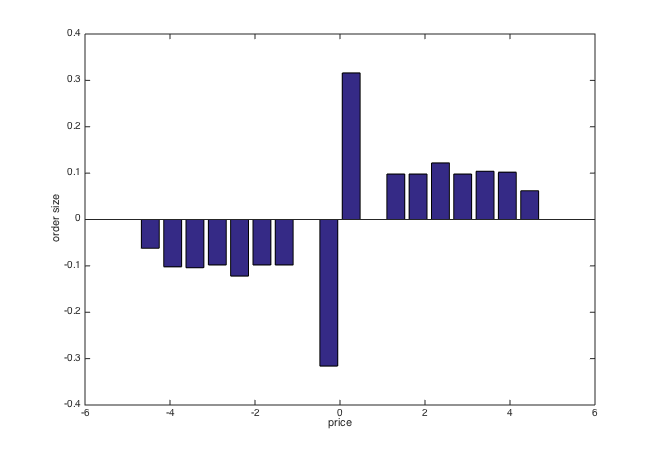}
    } & {
    \includegraphics[width = 0.49\textwidth]{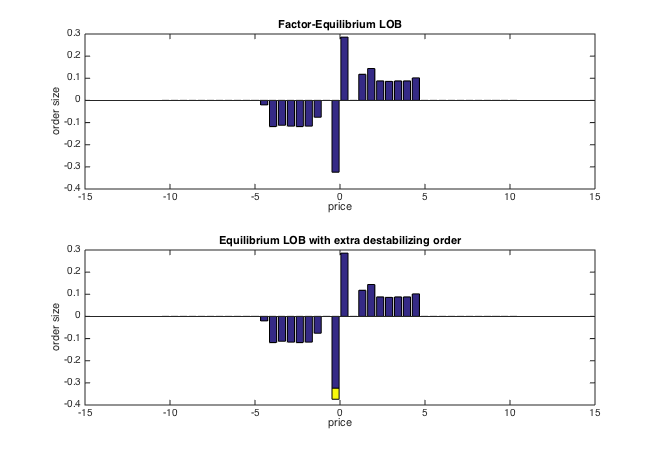}
    }\\
 \end{tabular}
  \caption{Left: LOB at time zero, with $\lambda^{+,\alpha_0}=2.5$, $\lambda^{-,\alpha_0}=1$, $\lambda^{+,\beta_0}=1$, $\lambda^{-,\beta_0}=2.5$. Right: equilibrium LOB at time zero, with the parameters depending on the market imbalance $I$ (top), and the same LOB, with an additional (yellow) limit order (bottom).}
    \label{fig:2}
\end{center}
\end{figure}

\begin{figure}
\begin{center}
  \begin{tabular} {cc}
    {
    \includegraphics[width = 0.49\textwidth]{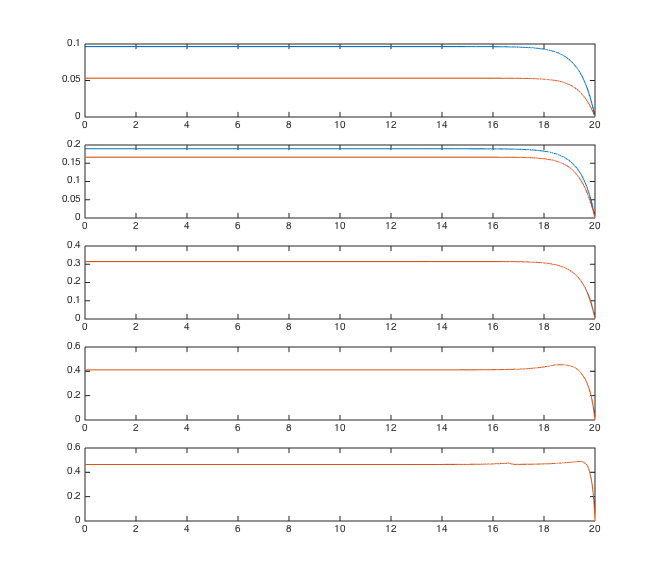}
    } & {
    \includegraphics[width = 0.49\textwidth]{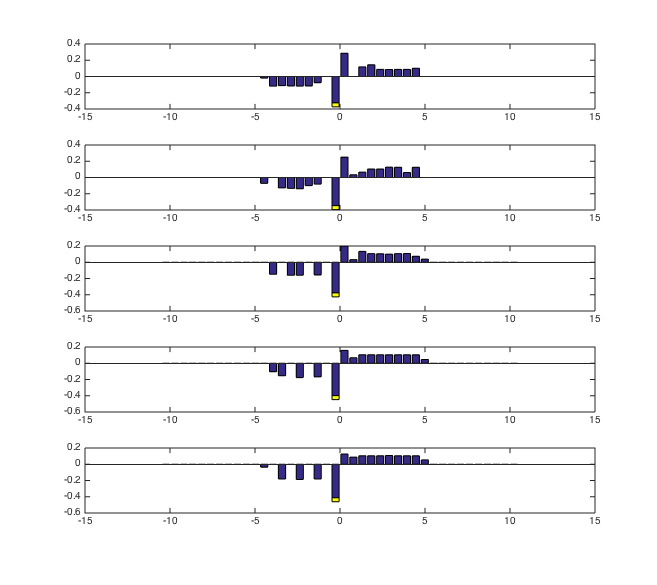}
    }\\
  \end{tabular}
  \caption{On the left: value functions $(V^b,V^a)$ (red and blue), as functions of time. On the right: LOB at each step of the convergence to a new equilibrium.}
    \label{fig:3}
\end{center}
\end{figure}

\newpage

\bibliographystyle{abbrv}
\bibliography{MFGLOB_cont_refs}






\end{document}